\documentclass[letterpaper,11pt]{article}

\usepackage{amsthm}
\usepackage{amsfonts}
\usepackage{amsmath}
\usepackage{amssymb}
\usepackage{graphicx}
\usepackage{epstopdf}
\usepackage{lineno}
\usepackage{setspace} 
\usepackage{verbatim}
\usepackage{hyperref}
\usepackage{authblk}
\usepackage{color}

\newtheorem{theorem}{Theorem}
\newtheorem*{theorem*}{Theorem}
\newtheorem{conjecture}{Conjecture}

\newtheorem{corollary}[theorem]{Corollary}
\newtheorem{claim}{Claim}
\newtheorem{lemma}[theorem]{Lemma}

\newtheorem{dfn}{Definition}

\newtheorem{proposition}[theorem]{Proposition}

\newcommand{\1}{\mathbf{1}}

\newcommand{\fs}{\mathcal S}

\newcommand{\fp}{\mathcal P}
\newcommand{\rhoo}{\rho_\1}
\newcommand{\leftover}{\mathit{leftover}\hspace{-.07em}}
\newcommand{\capacity}{\mathit{capacity}}
\newcommand{\prob}{\mathit{prob}}

\newenvironment{enumerate*}{
\begin{enumerate}
  \setlength{\itemsep}{5pt}
  \setlength{\parskip}{0pt}
  \setlength{\parsep}{0pt}
}{\end{enumerate}}

\newenvironment{itemize*}{
\begin{itemize}
  \setlength{\itemsep}{5pt}
  \setlength{\parskip}{0pt}
  \setlength{\parsep}{0pt}
}{\end{itemize}}

\title{Bounding the fractional chromatic number of $K_\Delta$-free graphs}
\author{Katherine Edwards\thanks{Email: ke@princeton.edu.  Supported by an NSERC PGS Fellowship and a Gordon Wu Fellowship.}}
\affil{Department of Computer Science\\Princeton University, Princeton, NJ}
\author{Andrew D.~King\thanks{Email: adk7@sfu.ca.  Supported by an EBCO/Ebbich Postdoctoral Scholarship and the NSERC Discovery Grants of Pavol Hell and Bojan Mohar.}}
\affil{Departments of Mathematics and Computing Science\\Simon Fraser University, Burnaby, BC}
\begin{document}

\maketitle

\begin{abstract}
King, Lu, and Peng recently proved that for $\Delta\geq 4$, any $K_\Delta$-free graph with maximum degree $\Delta$ has fractional chromatic number at most $\Delta-\tfrac{2}{67}$ unless it is isomorphic to $C_5\boxtimes K_2$ or $C_8^2$.  Using a different approach we give improved bounds for $\Delta\geq 6$ and pose several related conjectures.  Our proof relies on a weighted local generalization of the fractional relaxation of Reed's $\omega$, $\Delta$, $\chi$ conjecture.
\end{abstract}

\section{Introduction}

In this paper we consider simple, undirected graphs, and refer the reader to \cite{westbook} for unspecified terminology and notation.  We also work completely within the rational numbers.

The idea of bounding the chromatic number $\chi$ based on the clique number $\omega$ and maximum degree $\Delta$ goes all the way back to Brooks' Theorem, which states that for $\Delta\geq 3$, any $K_{\Delta+1}$-free graph with maximum degree $\Delta$ has chromatic number at most $\Delta$.  More recently, Borodin and Kostochka conjectured that if $\Delta \geq 9$, then any $K_{\Delta}$-free graph with maximum degree $\Delta$ has chromatic number at most $\Delta-1$ \cite{borodink77}.  The example of $C_5\boxtimes K_3$ (see Figure \ref{fig:upperbounds}) tells us that we cannot improve the condition that $\Delta\geq 9$.  Reed \cite{reed99} proved a weaker result that had been conjectured independently by Beutelspacher and Hering \cite{beutelspacherh84}:
\begin{theorem}\label{thm:reedbrooks}
For graph with $\Delta\geq 10^{14}$, if $\omega\leq \Delta-1$ then $\chi \leq \Delta-1$.
\end{theorem}
In the paper, Reed claims that a more careful analysis could replace $10^{14}$ with $10^3$.

\begin{figure}
\begin{center}
\includegraphics[scale=.3]{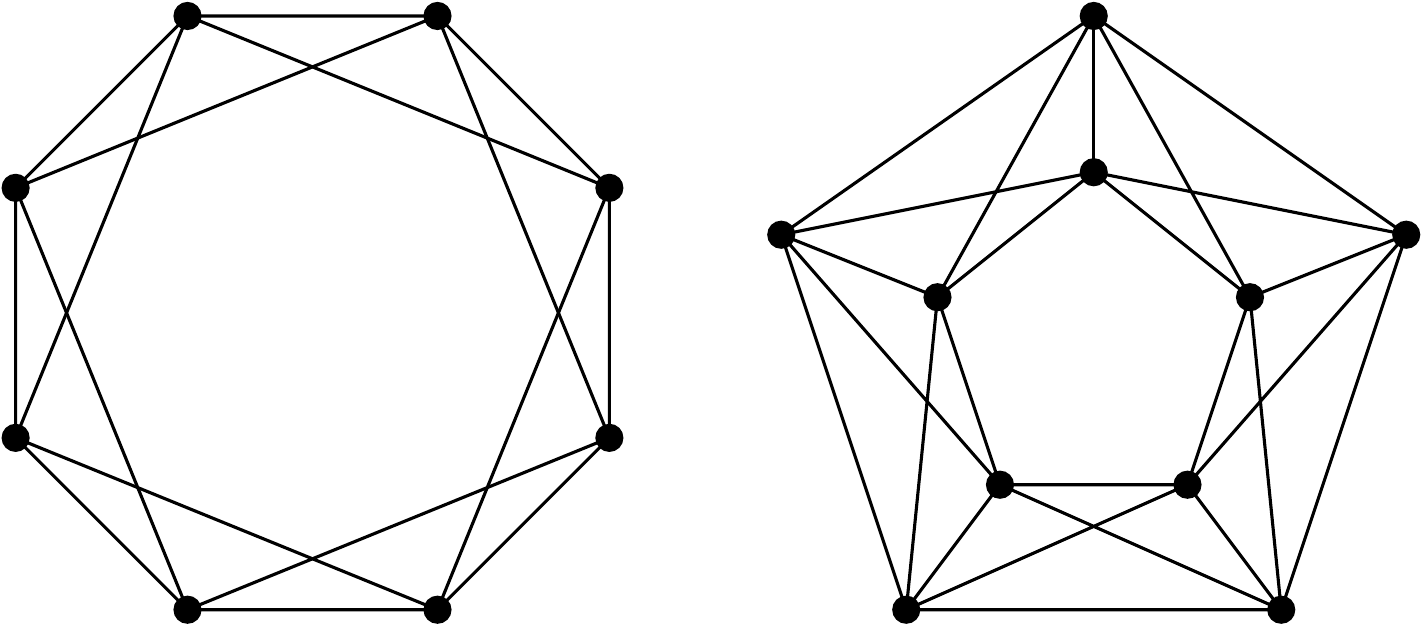}
\end{center}\caption{$C_8^2$ (left) and $C_5\boxtimes K_2$ (right).}\label{fig:counterexamples}
\end{figure}

This is the state of the art on the chromatic number of $K_\Delta$-free graphs, but what about the {\em fractional} chromatic number $\chi_f$ (we will define it soon) of $K_\Delta$-free graphs?  Albertson, Bollob\'{a}s, and Tucker noted in the 1970s that even when $\Delta\geq 3$, there are at least two $K_\Delta$-free graphs with $\chi_f=\Delta$, namely $C_8^2$ and $C_5\boxtimes K_2$ \cite{albertsonbt76} (see Figure \ref{fig:counterexamples}).  It turns out that these are the only such graphs.  For $\Delta\geq 3$ we define $f(\Delta)$ as:
$$f(\Delta)\ = \ \min_G\left\{\ \Delta-\chi_f(G) \ \mid \ \Delta(G) \leq \Delta;\ \ \omega(G)<\Delta;\ \ G\notin \{ C_8^2, C_5\boxtimes K_2 \} \ \right\}.$$

From Brooks' Theorem we know that $f(\Delta)$ is always nonnegative.  Considering Theorem \ref{thm:reedbrooks}, one may be inclined to believe that $f(\Delta)$ increases with $\Delta$.  As proven by King, Lu, and Peng, this is indeed the case for $\Delta \geq 4$ \cite{kinglp12}\footnote{For $\Delta\geq 6$, this is a consequence of the fact that when $\omega > \frac 23(\Delta+1)$, there is a stable set hitting every maximum clique \cite{king11}. For $\Delta \in \{4,5\}$ more work is required.}.  In Table \ref{tab:1} we show the known and conjectured bounds for various values of $\Delta$.  Figure \ref{fig:upperbounds} shows graphs demonstrating the best known (and conjectured) upper bounds on $f(\Delta)$ for $3\leq \Delta\leq 8$.

\begin{figure}[!h]
\begin{center}
\includegraphics[scale=.3]{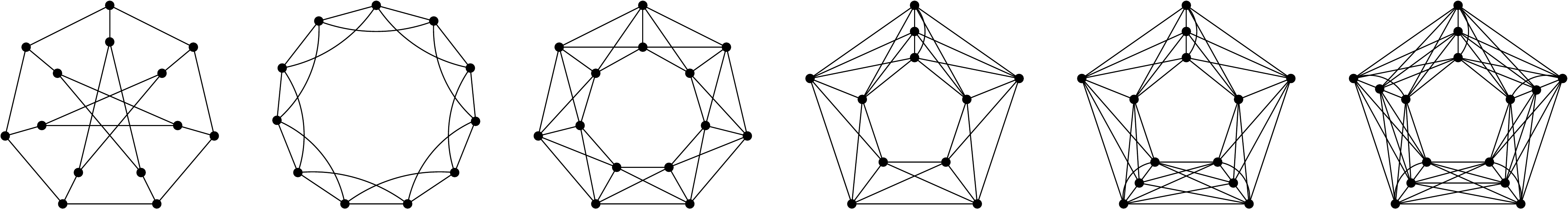}
\end{center}\caption{From left to right, the graphs $P(7,2)$, $C_{11}^2$, $C_7\boxtimes K_2$, $(C_5\boxtimes K_3)-4v$, $(C_5\boxtimes K_3)-2v$, and $C_5\boxtimes K_3$.}\label{fig:upperbounds}
\end{figure}

\begin{table}
\begin{center}
\begin{tabular}{|c|ccc|ccc|cc|}
\hline
&  \multicolumn{3}{|c|}{$f(\Delta)$} & \multicolumn{3}{|c|}{$f(\Delta)$} & \multicolumn{2}{|c|}{conjectured}\\
$\Delta$ &  \multicolumn{3}{|c|}{lower bounds} & \multicolumn{3}{|c|}{upper bound} &  \multicolumn{2}{|c|}{value}\\
\hline
3 & $3/64$ & $0.0468$ &\cite{hatamiz09}&  &  & & & \\
3 & $3/43$ & $0.0697$ &\cite{lupeng10}&  &  & & & \\
3 & $1/11$ & $0.0909$ &\cite{fergusonkk12}&  &  & & & \\
3 & $2/15$ & $0.1333$ &\cite{liu12}& $1/5$ & $P(7,2)$ & \cite{fajtlowicz77}& $1/5$ &\cite{heckmant01}\\
\hline
4 & $2/67$ & $0.0298$ &\cite{kinglp12}& $1/3$ & $C_{11}^2$ && $1/3$ &\cite{kinglp12} \\
5 & ${2/67}$ & ${0.0298}$ &\cite{kinglp12}& $1/3$ &$C_7\boxtimes K_2$ && $1/3$ &\cite{kinglp12}\\
6 & ${1/22.5}$ & $\mathbf{0.0445}$ && $1/2$ & $(C_5\boxtimes K_3)-4v$&& $\mathbf{1/2}$&\\
7 & ${1/11.2}$ & $\mathbf{0.0899}$ && $1/2$ & $(C_5\boxtimes K_3)-2v$&& $\mathbf{1/2}$&\\
8 & ${1/8.9}$ & $\mathbf{0.1135}$ && $1/2$ & $C_5\boxtimes K_3$ & \cite{catlin79}& $\mathbf{1/2}$& \\
9 &  ${1/7.7}$ & $\mathbf{0.1307}$ && $1$ & $K_8$ & & $1$ &\cite{borodink77}\\
10 & ${1/7.1}$ & $\mathbf{0.1423}$ && $1$ & $K_9$ & & $1$ &\cite{borodink77}\\
1000 & $1$ & $1$ &\cite{reed99}& $1$ & $K_{999}$&& $1$ &\cite{beutelspacherh84}\\
\hline
\end{tabular}
\end{center}
\caption{The state of the art.  New results and conjectures are in boldface.  For $\Delta \leq 5$, the fractional bound is the proven bound.  For $\Delta\geq 6$, the decimal bound approximates the proven bound, and the fractional expression approximates the decimal bound for ease of comparison.}
\label{tab:1}
\end{table}

In this paper we give improved bounds on $f(\Delta)$ for $\Delta \geq 6$ up until whenever Theorem \ref{thm:reedbrooks} takes effect, which we assume to be $\Delta=1000$.  We also conjecture that the upper bound of $f(\Delta)\leq \frac 12$ is tight for $\Delta\in \{6,7,8\}$:

\begin{conjecture}\label{con:1}
For $\Delta\in \{6,7,8\}$, let $G$ be a graph with maximum degree $\Delta$ and clique number at most $\Delta-1$.  Then 
the fractional chromatic number of $G$ is at most $\Delta-\frac 12$.
\end{conjecture}

One of the major questions in this area, as is evident from Table \ref{tab:1}, is the following:

\begin{conjecture}\label{con:2}
For $\Delta\geq 3$, $f(\Delta)\leq f(\Delta+1)$.
\end{conjecture}

\section{Fractionally colouring weighted and unweighted graphs}

In this paper we must consider fractional colourings of both vertex-weighted and unweighted graphs, because we will begin to fractionally colour an unweighted graph $G$ in one way that does very well on particularly tricky vertices, then finish the colouring in another way that does fairly well on all vertices.  The second step requires a weighted generalization of a known result; the weight on a vertex reflects how much colour we have yet to assign to the vertex.

Let $G=(V,E)$ be a graph, let $\fs = \fs(G)$ be the set of stable sets of $G$, and let $k$ be a nonnegative rational.  Now let $\kappa: \fs\rightarrow \fp([0,k))$ be a function assigning each stable set $S$ of $G$ a subset of $[0,k)$ such that for every $S\in \fs$, $\kappa(S)$ is the union of disjoint half-open intervals\footnote{(containing their lower endpoint but not their upper)} with rational endpoints between $0$ and $k$, and for any distinct $S,S'$ in $\fs$, $\kappa(S)\cap \kappa(S') = \emptyset$.  For a set $\fs'\subseteq \fs$ of stable sets, define $\kappa(\fs')$ as $\cup_{S\in \fs'}\kappa(S)$.  For each $v\in V$, define $\kappa[v]$ as $\cup_{S\ni v}\kappa(S)$.  For a set $X\subseteq V$, define $\kappa[X]$ as $\cup_{S\cap X\neq \emptyset}\kappa(S) = \cup_{v\in X}\kappa[v]$.

Now consider a nonnegative vertex weight function $w: V\rightarrow [0,\infty)$; in this case we say that $G$ is a {\em $w$-weighted graph}.  (Recall that $w$, like all numbers considered in this paper, is rational.)  If for every vertex $v\in V$ we have $|\kappa[v]|\geq w(v)$, then $\kappa$ is a {\em fractional $\wr w$-colouring} of $G$ with {\em weight} $k$; in other words it is a {\em fractional $k\wr w$-colouring} of $G$.  The minimum weight of a fractional $\wr w$-colouring $G$ is denoted $\chi_f^w(G)$, or simply $\chi_f^w$ when the context is clear.  If $w = \1$ (i.e.\ the weight function uniformly equal to 1), then we may omit it from the notation, i.e.\ we define fractional colourings and the fractional chromatic number of unweighted graphs.  If some vertex $v$ has $|\kappa[v]| < w(v)$, we say that we have a {\em partial fractional $k\wr w$-colouring} of $G$.

In both settings, $\kappa[v]$ is the {\em colour set} assigned to $v$.  We denote the colours {\em available to $v$} (i.e.\ not appearing on the neighbourhood of $v$) by $\alpha(v)$, that is, $\alpha(v) = [0,k)\setminus \kappa[N(v)]$.

This is just one of several ways to think about fractional colourings; we hold the following proposition to be self-evident\footnote{The unweighted version is described as folklore in \cite{fergusonkk12} and was used earlier in \cite{kaiserkk10}, and probably elsewhere.}: 

\begin{proposition}\label{prop:tfae}
Let $G$ be a $w$-weighted graph.  The following are equivalent:
\begin{enumerate}
\item[(1)] $G$ has a fractional $k\wr w$-colouring.
\item[(2)] There is an integer $c$ and a multiset of $ck$ stable sets of $G$ such that every vertex $v$ is contained in at least $c\cdot w(v)$ of them.
\item[(3)] There is a probability distribution on $\fs$ such that for each $v\in V$, given a stable set $S$ drawn from the distribution, $\Pr(v\in S) \geq w(v)/k$.
\end{enumerate}
\end{proposition}

For more background on fractional colourings we refer the reader to \cite{scheinermanubook}.  At this point it is convenient to prove a useful consequence of Hall's Theorem that we will use repeatedly in Section \ref{sec:mce}:
\begin{lemma}\label{lem:hall}
Let $\kappa$ be a partial fractional $k\wr w$-colouring of $G$, and let $X$ be the set of vertices $v$ with $|\kappa[v]|<w(v)$.  Suppose for every $X'\subseteq X$ we have
\begin{equation}\label{eq:hall}\left | \bigcup_{v\in X'}\alpha(v) \right | \geq  \sum_{v\in X'} w(v).\end{equation}
Then there is a fractional $k\wr w$-colouring of $G$.
\end{lemma}

\begin{proof}
We may assume (by uncolouring $X$) that for every $v\in X$, $\kappa[v]=\emptyset$.  Thus we have a fractional $k\wr w$-colouring of $G-X$.  By Proposition \ref{prop:tfae} there is an integer $c$ and a multiset of $ck$ stable sets $S_1,\ldots,S_{ck}$ of $G-X$ such that every vertex $v \notin X$ is in at least $c\cdot w(v)$ of them.

We now set up Hall's Theorem by constructing a bipartite graph $H$ with vertex set $A\cup B$.  Let $A$ consist of, for every $v\in X$, $c\cdot w(v)$ copies of $v$.  Let $B$ consist of vertices $b_1,\ldots b_{ck}$.  For every vertex $a$ of $A$, let $a$ be adjacent to $b_i$ if and only if the vertex $v$ in $X$ corresponding to $a$ has no neighbour in $S_i$.  Equation (\ref{eq:hall}) guarantees that for every $A'\subseteq A$, $|N(A')| \geq |A'|$, so by Hall's Theorem we have a matching in $H$ saturating $A$.  This matching corresponds to a partial mapping $m:[ck]\rightarrow X$ such that
\begin{itemize}
\item for every $i\in [ck]$ in the domain of $m$, $S_i\cup m(i)$ is a stable set, and
\item for every $v\in X$, at least $c\cdot w(v)$ elements of $[ck]$ map to $v$.
\end{itemize}
Thus we can extend the stable sets $S_i$ appropriately; by Proposition \ref{prop:tfae}, this gives the desired fractional $k\wr w$-colouring of $G$.
\end{proof}

We remark that this lemma is most sensibly applied when $X$ is a clique.

\subsection{Reed's Conjecture and fractional colourings}

Our approach to fractionally colouring $K_\Delta$-free graphs is inspired by the following result of Reed (\cite{molloyrbook}, \S 21.3):

\begin{theorem}\label{thm:reed}
Every graph $G$ satisfies $\chi_f(G) \leq \frac 12(\Delta(G)+1+\omega(G))$.
\end{theorem}

This is the fractional relaxation of Reed's $\omega$, $\Delta$, $\chi$ conjecture \cite{reed98}, which proposes that every graph satisfies $\chi \leq \lceil \frac 12(\Delta+1+\omega)\rceil$.  However, we do not consider the conjecture, or even the fractional relaxation, but rather a weighted version of a local strengthening observed by McDiarmid (\cite{molloyrbook}, p.246).  For a vertex $v$ let $\omega(v)$ be the size of the largest clique containing $v$.  Then:

\begin{theorem}\label{thm:mcdiarmid}
Every graph $G$ satisfies $\chi_f(G) \leq \max_v \frac 12(d(v)+1+\omega(v))$.
\end{theorem}

The proof of this theorem was never published, but appears in Section 2.2 of \cite{kingthesis} and is almost identical to the proof of Theorem \ref{thm:reed}.  What we need is a new weighted version of this theorem, which we prove here.  First we need some notation.  For a vertex $v$ let $\tilde N(v)$ denote the closed neighbourhood of $v$.  Given a $w$-weighted graph $G$ and a vertex $v \in V(G)$, we define:
\begin{itemize}
\item The {\em degree weight} $w_d(v)$ of $v$, defined as $\sum_{u\in \tilde N(v)}w(u)$.
\item The {\em clique weight} $w_c(v)$ of $v$, defined as the maximum over all cliques $C$ containing $v$ of $\sum_{u\in C}w(u)$.
\item The {\em Reed weight} $\rho_w(v)$ of $v$, defined as $\frac 12(w_d(v)+w_c(v))$ (we sometimes denote $\rhoo$ by $\rho$).  For a graph $G$, we define $\rho_w(G)$ as $\max_{v\in V(G)}\rho_w(v)$.
\end{itemize}

Our result is a natural generalization of McDiarmid's:
\begin{theorem}\label{thm:weighted}
Every graph $G$ satisfies $\chi^w_f(G) \leq \rho_w(G)$.
\end{theorem}

\begin{proof}
Let $c$ be a positive integer such that for every $v$, $c w(v)$ is an integer; $c$ exists since the weights are rational.  Let $G_w$ be the graph constructed from $G$ by replicating each vertex $v$ into a clique $C_v$ of size $cw(v)$.~\footnote{That is, $x\in C_u$ and $y\in C_v$ are adjacent precisely if $u,v$ are adjacent or if $u=v$ and $x,y$ are distinct.}  Applying Theorem \ref{thm:mcdiarmid} to $G_w$ tells us that there is a fractional $c\rho_w(G)$-colouring $\kappa_w$ of $G_w$.  From this we construct a  $cw$-fractional $c\rho_w(G)$-colouring $\kappa$ of $G$ by setting, for each $v\in V(G)$,
$$\kappa[v] = \kappa_w[C_v].$$
The result follows from Proposition \ref{prop:tfae} (3).
\end{proof}

\section{The general approach}

Fix some $\Delta \geq 6$ and $0< \epsilon\leq  \frac 12$, and suppose we wish to prove that $f(\Delta)\geq \epsilon$.  Let $G$ be a graph with maximum degree $\Delta$ and clique number $\omega \leq \Delta-1$; by Theorem \ref{thm:reed} we know that $\chi_f(G)\leq \Delta-\frac 12$ if $\omega \leq \Delta-2$, so we assume $G$ has clique number $\omega=\Delta-1$.  We define $V_\omega$ as the set of vertices in $\omega$-cliques, and $V'_\omega$ as the set of vertices in $V_\omega$ with degree $\Delta$.  Let $G_\omega$ and $G'_\omega$ denote the subgraphs of $G$ induced on $V_\omega$ and $V'_\omega$ respectively.  Notice that a vertex $v$ will have $\rhoo(v)> \Delta-\frac 12$ if and only if $v$ is in $V'_\omega$.  In plain language, our approach is:
\begin{enumerate}
\item Prove that in a minimum counterexample, $G_\omega$ has a nice structure.
\item Spend a little bit of weight on a fractional colouring that lowers the Reed weight for vertices in $V'_\omega$ at a rate of $(1+\epsilon')$ per weight spent, i.e.\ we spend $y$ weight and $(1+\epsilon')y = y+\epsilon$.  If $y$ is sufficiently small, this lowers the maximum Reed weight over all vertices of $G$ by $y+\epsilon$.
\item Having already ``won'' by $\epsilon$, i.e.\ having lowered $\rho(G)$ by $y+\epsilon$ using only $y$ colour weight, we can finish the colouring using Theorem \ref{thm:weighted}.
\end{enumerate}
More specifically, we find a vertex weighting $w$ such that we have a fractional $y\wr w$-colouring of $G$, and such that $\rho_{(\1-w)}(G) \leq \Delta-y-\epsilon$.  We then apply Theorem \ref{thm:weighted} to find a fractional $(\Delta-y-\epsilon) \wr (\1-w)$-colouring of $G$.  Combining these two partial fractional colourings gives us a fractional $(\Delta-\epsilon)$-colouring of $G$.

Since any $v\notin V_\omega'$ satisfies $\rho_{\1}(v)\leq \Delta-\frac 12$, if $(1+\epsilon')y \leq \frac 12$ we only need to ensure that $\rho$ drops by $(1+\epsilon')y$ for vertices with $\rhoo(v) = \Delta$.  Actually we can ensure that while we do this, $\rho$ also drops at a decent rate (easily at least $\frac 12y$) for vertices with $\rho < \Delta$.  This means that we can spend more weight (i.e.\ increase $y$), thereby improving $\epsilon$.  It is in our interests to first worry about maximizing $\epsilon'$, then worry about maximizing $y$.

This method depends heavily on properly understanding the structure of vertices with $\rhoo(v)=\Delta$.  We simplify this structure through reductions, or if you prefer, the structural characterization of a minimum counterexample:

\begin{lemma}\label{lem:mce}
Fix some $\Delta \geq 5$ and some $\epsilon \leq \frac 12$, with the further restriction that $\epsilon \leq \frac 13$ if $\Delta = 5$.  Let $G$ be a graph with maximum degree $\Delta$ and clique number at most $\Delta-1$ such that
\begin{itemize*}
\item if $\Delta = 5$, no component of $G$ is isomorphic to $C_{5}\boxtimes K_2$,
\item $G$ has fractional chromatic number greater than $\Delta-\epsilon$, and
\item no graph on fewer vertices has these properties.
\end{itemize*}
Then
\begin{enumerate*}
\item[(i)] the maximum cliques of $G$ are pairwise disjoint, and
\item[(ii)] there is no vertex $v$ outside a maximum clique $C$ such that $|N(v)\cap C| > 1$.
\end{enumerate*}
\end{lemma}

Together, these properties allow us to apply the following result of Aharoni, Berger, and Ziv \cite{aharonibz07}:

\begin{theorem}\label{thm:aharoni}
Let $k$ be a positive integer and let $G$ be a graph whose vertices are partitioned into cliques of size $\omega\geq 2k$.  If $G$ has maximum degree at most $\omega+k-1$, then $\chi_f(G) = \omega$.
\end{theorem}
Applying this theorem to an induced subgraph of $G_\omega$ is the key to proving that we can lower $\rho$ quickly for any vertex $v$ with $\rhoo(v)=\Delta$.  The proof of Lemma \ref{lem:mce} is technical, independent of the main proof, and does not give insight to our approach, so we defer it to Section \ref{sec:mce}.  We now consider the probability distribution on stable sets that, via Proposition \ref{prop:tfae}, characterizes our initial colouring phase.

From now until Section \ref{sec:mce}, we consider $G$ to be a graph with maximum degree $\Delta \geq 6$, clique number $\omega=\Delta-1$, and satisfying properties {\em (i)} and {\em (ii)} of Lemma \ref{lem:mce}. We remark that Lemma \ref{lem:mce} gives a characterization of minimum counterexamples with $\Delta=5$; although we do not make use of the characterization in this paper, it is likely to be useful in the future.

\section{A probability distribution}\label{sec:prob}

For every vertex $v$ of $G$, let $N_\omega(v)$ denote $N(v)\cap V_\omega$ and let $d_\omega(v)$ denote $|N_\omega(v)|$.  The initial phase of our colouring involves choosing a random stable set $S_w$ of $G_w$, then extending it randomly to a stable set $S$ of $G$ such that $S_w$ and $S$ have the following desirable properties:

\begin{enumerate}
\item For every $v\in V_\omega$,
\begin{equation}
\Pr(v\in S_\omega) = \tfrac 1\omega.
\label{eq:one}
\end{equation}
\item For every $v\notin V_\omega$,
\begin{eqnarray}\label{eq:two}
\Pr(N_\omega(v)\cap S_\omega = \emptyset) &\geq& \sum_{i=0}^3\tfrac 14\Pr\left(\mathrm{Bin}(d_\omega(v),\tfrac 4\omega ) \leq i )\right)\\
&=& \sum_{i=0}^3\tfrac{4-i}4\Pr\left(\mathrm{Bin}(d_\omega(v),\tfrac 4\omega ) = i )\right).\nonumber
\end{eqnarray}
\item For every $v\notin V_\omega$,
\begin{equation}
\Pr(v\in S) \geq \frac{\Pr(N_\omega(v)\cap S_\omega = \emptyset)}{(d(v)-d_\omega(v))+1} \geq \frac{\sum_{i=0}^3\tfrac{4-i}4\Pr\left(\mathrm{Bin}(d_\omega(v),\tfrac 4\omega ) = i )\right)}{(d(v)-d_\omega(v))+1}.
\label{eq:three}
\end{equation}
\item $S$ is maximal.
\end{enumerate}

We will put weight on stable sets according to this distribution until we can no longer guarantee that $\rho$ is dropping quickly.  We discuss this stopping condition in Section \ref{sec:stopping}.

\subsection{Choosing $S_\omega$}

Denote the maximum cliques of $G$ by $B_1,\ldots,B_\ell$, bearing in mind that they are vertex-disjoint.  To choose $S_\omega$ we first select, for each $1\leq i\leq \ell$, a subset $B'_i$ of $B_i$ of size 4, uniformly at random and independently for each $i$.  Setting $\tilde G_\omega$ to be the subgraph of $G$ induced on $\cup_i B'_i$, note that every vertex in $B_i$ has at most two neighbours outside $B_i$ and therefore $\Delta(\tilde G_\omega) \leq 5$.  Thus Theorem \ref{thm:aharoni} tells us that $\tilde G_\omega$ is fractionally $4$-colourable.  It follows from Proposition \ref{prop:tfae} that there is a probability distribution on the stable sets of $\tilde G_\omega$ such that given a stable set $\tilde S$ chosen from this distribution, for any $v\in \tilde G_\omega$, $\Pr(v\in \tilde S)=\tfrac 14$.

We therefore choose $S_\omega$ from this distribution, subject to our random choice of $\tilde G_\omega$.  Since every $v\in G_\omega$ satisfies $\Pr(v\in \tilde G_\omega) = \tfrac 4\omega$, for any $v\in G_\omega$ we clearly have $\Pr(v\in S_\omega)=\tfrac 1\omega$, i.e. (\ref{eq:one}) holds.  We must now prove that (\ref{eq:two}) holds (the reader may have noticed that any old fractional $\omega$-colouring of $G_\omega$ would have given us $S_\omega$ satisfying (\ref{eq:one})).

The first step is to observe that for $v \notin G_\omega$ and $0 \leq i \leq 3$,

\begin{equation}\label{eq:twoa}
\Pr\left((N_\omega(v)\cap S_\omega = \emptyset) \mid (|N_\omega(v)\cap \tilde G_\omega| = i)\right) \geq \frac{4-i}{4}.
\end{equation}
This is because every neighbour of $v$ in $\tilde G_\omega$ is in $S_w$ with probability $\frac 14$, and in the worst case these events may be disjoint for all $i$ such neighbours (we later conjecture that it is possible to avoid this worst case; this would improve our bounds substantially for $\Delta\in \{5,6\}$).

The second step is to observe that for $v\notin G_\omega$ and $0\leq i \leq d_\omega(v)$,
\begin{equation}\label{eq:twob}
\Pr\left(|N_\omega(v)\cap \tilde G_\omega| = i \right) = \Pr\left(\mathrm{Bin}(d_\omega(v),\tfrac 4\omega ) = i )\right).
\end{equation}
To see this, note that Lemma \ref{lem:mce} tells us that any two neighbours $x,y\in G_\omega$ of $v$ are in different blocks $B_i$, and therefore the events of $x$ being in $\tilde G_\omega$ and $y$ being in $\tilde G_\omega$ are independent.  Equation (\ref{eq:two}) follows immediately from Equations (\ref{eq:twoa}) and (\ref{eq:twob}).

\subsection{Choosing $S$}

Given a choice of $S_\omega$, we randomly extend to $S$ as follows:
\begin{enumerate}
\item Choose an ordering $\pi$ of $V(G)\setminus V_\omega$ uniformly at random, and label the vertices of $V(G)\setminus V_{\omega}$ as $v_1,\ldots,v_r$ in the order in which they appear in $\pi$.
\item Set $S = S_\omega$.
\item For each of $i = 1,\ldots, r$ in order, put $v_i$ in $S$ if and only if it currently has no neighbour in $S$.
\end{enumerate}

Since every vertex in $V_\omega$ is in $S_\omega$ or has a neighbour in $S_\omega$, and every vertex not in $V_\omega$ is in $S$ or has a neighbour in $S$, we can see that $S$ is always a maximal stable set.  A vertex $v_i\in V(G)\setminus V_\omega$ is in $S$ if it has no neighbours in $S_\omega$, and it is not adjacent to any $v_j\in V(G)\setminus V_\omega$ for $j<i$.  Since we choose $\pi$ uniformly at random, any vertex $v\in V(G)\setminus V_\omega$ satisfies
\begin{equation}\label{eq:threea}
\Pr\left((v\in S) \mid (N_\omega(v)\cap S_\omega = \emptyset)   \right) \geq \frac{1}{|N(v)\setminus V_\omega|+1}.
\end{equation}
Equation (\ref{eq:three}) follows immediately from Equation (\ref{eq:threea}).

\subsection{Bounding the rate at which $\rho$ initially decreases}

Suppose we spend weight $y$ to colour $G$ according to the probability distribution on $S$ that we just described.  That is, for $S'\in \fs(G)$, we place weight $q(S')$ on $S'$, where
$$q(S') = y \cdot \Pr(S=S').$$
Then we wish to argue that $\rho(G)$ drops by $(1+\epsilon')y$ for some positive $\epsilon'$.  For now, to avoid consideration of stopping conditions\footnote{i.e.\ when $y$ is large enough to make our model fail}, suppose that $y$ is very small ($y=\frac 1{10}$ will do for now).

For a fixed $\Delta$ and $0\leq d\leq \Delta$ we define $p(\Delta,d)$ as
\begin{equation}
p(\Delta,d) = \frac{\sum_{i=0}^3\tfrac 14\Pr\left(\mathrm{Bin}(d,\tfrac 4{\omega} ) \leq i )\right)}{(\Delta-d)+1},
\end{equation}
noting that a vertex $v\notin G_\omega$ with $d_\omega(v) = d$ is in $S$ with probability at least $p(\Delta,d)$.  Following this, we define
$$
\mu_k(\Delta) = \min_{0\leq d\leq k} p(\Delta,d)   \qquad \mathrm{and} \qquad \mu(\Delta) = \mu_\Delta(\Delta) = \min_{0\leq d\leq \Delta} p(\Delta,d),
$$
%
noting that any vertex $v\notin G_\omega$ is in $S$ with probability at least $\mu(\Delta)$.

\begin{lemma}
For every vertex $v\in V(G)$, $\Pr(v\in S)\geq \mu(\Delta)$.
\end{lemma}
\begin{proof}
To see this we only need to prove that $v\in G_\omega$ is in $S$ with probability at least $\mu(\Delta)$.  This is clearly the case since $v$ is in $S$ with probability $\frac 1\omega > \frac1{\Delta+1} = p(\Delta,0) \geq \mu(\Delta)$.
\end{proof}

We now set $\epsilon'$ to be $\mu(\Delta)$.  Table \ref{tab:mu} gives some computed values of $\mu(\Delta)$, and Figure \ref{fig:plots} shows some values of $p(\Delta,d)$.  (We will define and consider $\tilde y(\Delta)$ in the next section.)  These numbers were computed using Sage; the code is available at \cite{sagekdeltafree}.

\begin{table}
\begin{center}
\begin{tabular}{|c|c|c|c|c|c|}
\hline
$\Delta$ & $\mu(\Delta)$ & $\mu(\Delta)(\Delta+1)$&$d$ for which $\mu(\Delta)=p(\Delta,d)$ & $\tilde y(\Delta)$ & $\tilde y(\Delta)\mu(\Delta)$\\
\hline
6 & .029376 & .205 & 6 & 1.518 & 0.04459\\
7 & .054869 & .439 & 6 & 1.640 & 0.08999\\
8 & .062947 & .567 & 7 & 1.804 & 0.11353\\
9 & .066406 & .664 & 7 & 1.969 & 0.13077\\
10 & .066328 & .730 & 8 & 2.146 & 0.14234\\
100 & .009843 & .994 & 29 & 20.003 & 0.19691 \\
1000 & .000998 & .999 & 135 & 199.979 & 0.19973 \\
\hline
\end{tabular}
\end{center}
\caption{Some values of $\mu(\Delta)$, where they are achieved, and corresponding values of $\tilde y$, which we discuss later.  Note that $p(\Delta,0)=1/(\Delta+1)$ is an upper bound for $\mu(\Delta)$.  These values are calculated in \cite{sagekdeltafree}.}
\label{tab:mu}
\end{table}

\begin{figure}
\begin{centering}
\includegraphics[scale=.5]{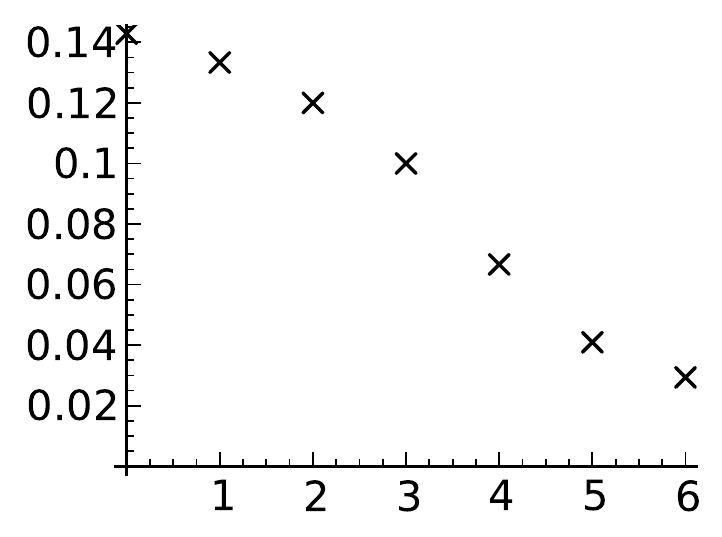}
\includegraphics[scale=.5]{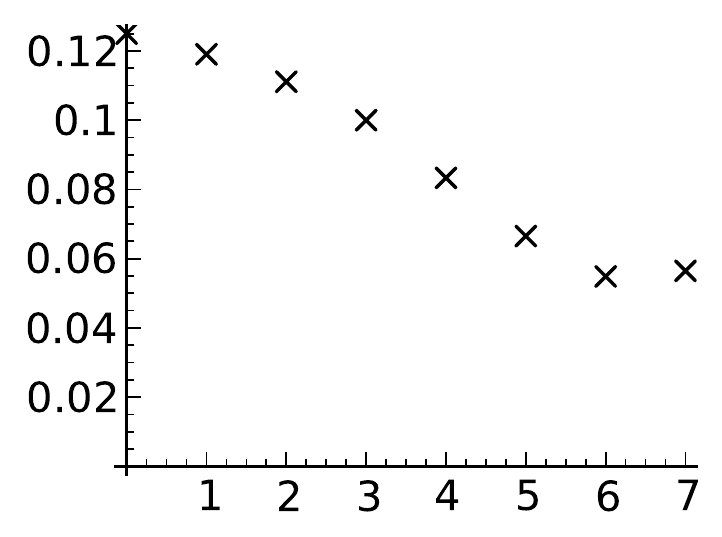}
\includegraphics[scale=.5]{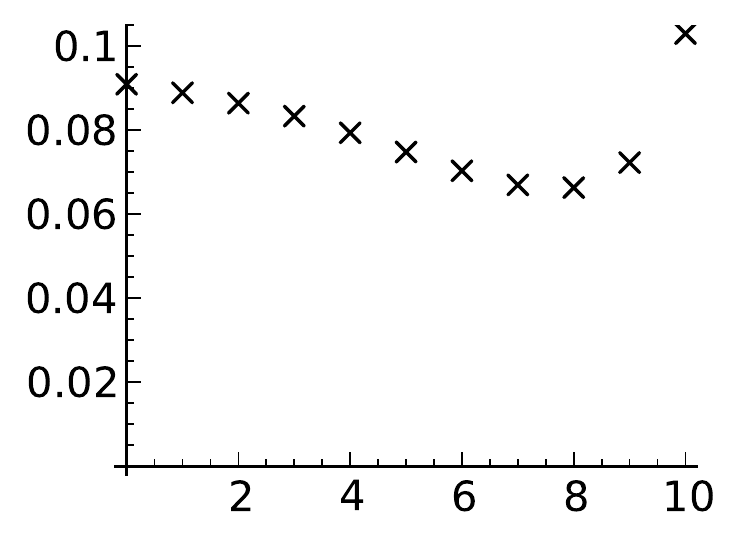}
\includegraphics[scale=.5]{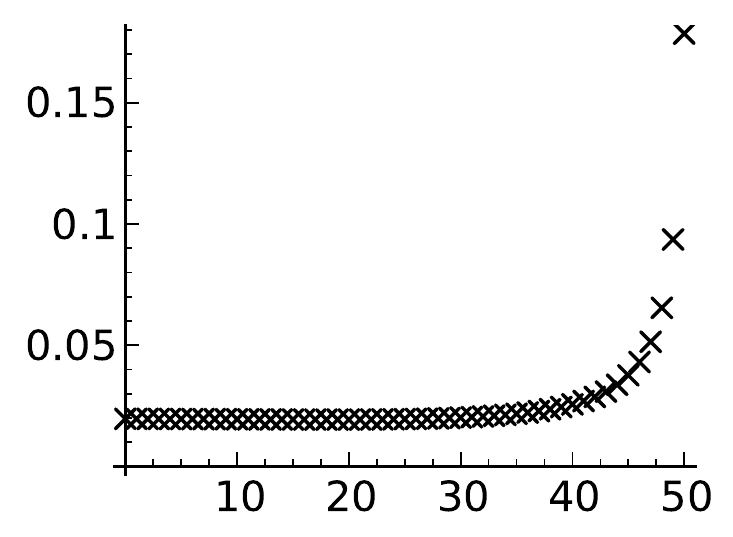}
\end{centering}\caption{Values of $d$ versus $p(\Delta,d)$ for $\Delta \in \{6,7,10,50\}$.}\label{fig:plots}
\end{figure}

\begin{lemma}
For any vertex $v$ in $V'_\omega$, $E(|S\cap \tilde N(v)|) \geq 1 + 2\epsilon'$.
\end{lemma}
\begin{proof}
Since $v$ is in some $B_i$ and has degree $\Delta = 1+\omega$, $v$ has exactly two neighbours outside $B_i$.  Each is in $S$ with probability at least $\epsilon'$, and $S$ contains a vertex in $B_i$ with probability $1$.  Therefore the lemma follows from linearity of expectation.
\end{proof}

Let $v$ be a vertex in $V'_\omega\cap B_i$.  Since $E(|S\cap B_i|)=1$, and $B_i$ is the unique maximum clique containing $v$, we know that at the outset, when we spend weight $y$, $\rho(v)$ will drop by $\frac 12(1 + 1 + 2\epsilon')y = (1+\epsilon')y$.

For $k \leq \omega$, let $V_k$ be the set of vertices in a clique of size $k$ but not a clique of size $k+1$, noting that these vertex sets partition $V(G)$.  We note the following.

\begin{lemma}\label{lem:bigclique}
If $4\leq k\leq  \omega-1$ and $v$ is a vertex in $V_k$, then $v$ has at most $\Delta+1-k$ neighbours in $V_\omega$.
\end{lemma}
\begin{proof}
It suffices to prove that if $X$ is a $k$-clique containing $v$, then $X$ does not intersect an $\omega$-clique.  Suppose it does intersect some $B_i$, and note that it may only intersect $B_i$ once by Lemma \ref{lem:mce}.  Since any vertex in $B_i$ has at most two neighbours outside $B_i$, $|X|$ must be at most 3, a contradiction.
\end{proof}

\begin{corollary}
If $v\in V_k$ for some $4\leq k\leq \omega-1$, then $\Pr(v\in S) \geq \mu_{\Delta+1-k}(\Delta)$.
\end{corollary}

\section{The initial colouring}

The probability distribution described in the previous section tells us what to do in the initial colouring phase: we choose colour classes according to the distribution.  The only thing we need to worry about is giving a vertex more than colour weight 1.  To avoid this, when a vertex is full we simply delete it and continue as though it never existed.  This is the same approach taken in the proof of Theorems \ref{thm:reed} and \ref{thm:mcdiarmid}.  Vertices in $V_\omega$ will never be full before the end of our process.

\begin{lemma}\label{lem:initial}
For any $y\in [0,\omega]$ there exists a vertex weighting $w$ and a fractional $y\wr w$-colouring of $G$ such that $w$ satisfies the following conditions:
\begin{enumerate}
\item[(a)] Every vertex $v$ in $V_\omega$ has $w(v) = y/\omega$.
\item[(b)] For $0\leq \ell \leq \Delta$, every vertex $v \notin V_\omega$ with exactly $\ell$ neighbours in $V_\omega$ has $w(v)\geq \min\{p(\Delta,\ell)y,1 \}$.
\item[(c)] For $1\leq k <\omega$, every clique $X$ of size $k$ has $w(X) \geq k\min\{\mu(\Delta)y, 1\}$.
\item[(d)] For $4\leq k<\omega$, every clique $X$ of size $k$ has $w(X) \geq k\min\{ \mu_{\Delta+1-k}(\Delta)y, 1\}$.
\item[(e)] Every vertex $v$ with $w(v)<1$ has $w(\tilde N(v))\geq y$.
\end{enumerate}
\end{lemma}
Note that $\mu(\Delta)y$ and $\mu_{\Delta+1-k}(\Delta)y$ are less than 1.
\begin{proof}
We proceed using the following algorithm.

\noindent Initially, set $H_0=G$, set $\leftover_0 = y$, and set $\capacity_0(v)=1$ for every vertex in $H_0$.  For $i=0,1,\ldots$ do the following.
\begin{enumerate}
\item Let $R_i$ be a random stable set drawn from the distribution giving $S$ described in Section \ref{sec:prob}.  For every vertex $v$ we set $\prob_i(v)$ as $\Pr(v\in R_i)$.
\item Set $y_i'$ to be $\min_{v\in V(H_i)}(\capacity_i(v)/\prob_i(v))$, and set $y_i$ to be $\min\{\leftover_i, y_i' \}$.
\item For every $v\in V(H_i)$, set $w_i(v)$ to be $\prob_i(v)y_i$.
\item For every $v\in V(H_i)$, set $\capacity_{i+1}(v)$ to be $\capacity_i(v)-w_i(v)$.
\item Set $\leftover_{i+1}$ to be $\leftover_i-y_i$.
\item If $\leftover_{i+1} = 0$, we terminate the process.  Otherwise, let $U_i$ be the vertex set $\{v\in V(H_i) \mid \capacity_{i+1}(v)=0\}$, and set $H_{i+1}$ to be $H_i-U_i$.
\end{enumerate}
Let $\nu$ denote the value of $i$ for which $\leftover_{i+1}=0$.  For every vertex $v$, let $w(v) = \sum_{i=0}^\nu w_i(v)$.  Observe that $y=\sum_{i=0}^\nu y_i$.\\

We first prove that this process must terminate.  Our choice of each $y_i$ implies that either $\leftover_{i+1}=0$, or $|U_{i+1}|<|U_i|$.  Thus we terminate after at most $|V(G)|$ iterations.  Now observe that every vertex $v\in G_\omega$ has $\prob_i(v) = 1/\omega$ throughout the process, and therefore $\capacity_\nu(v) > 0$ since $\leftover_0 = y\leq \omega$ (this can easily be proved by induction on $i$).  Note that (a) also follows from this observation.  As a further consequence, we can see that $G_\omega$ is a subgraph of every $H_i$.

We claim that we actually have a collection of fractional $y_i\wr w_i$-colourings for $0\leq i\leq \nu$.  To see this we simply appeal to Proposition \ref{prop:tfae} (3), noting that $\Pr(v\in R_i) = w_i(v)/y_i$.  Since $w = \sum_{i=1}^\nu w_i$ and $y = \sum_{i=0}^\nu y_i$, it follows immediately that these colourings together give us a fractional $y\wr w$-colouring of $G$.

To prove (b), we take $v\notin V_\omega$ with $\ell$ neighbours in $V_\omega$, and assume that $w(v)<1$, otherwise we are done.  Since every $H_i$ contains $G_\omega$, we can see that

\begin{equation}
\Pr(v\in R_i) \geq \frac{\sum_{i=0}^3\tfrac{4-i}4\Pr\left(\mathrm{Bin}(d_\omega(v),\tfrac 4\omega ) = i )\right)}{|N(v)\cap V(H_i)|-d_\omega(v)+1} \geq p(\Delta,\ell).
\end{equation}
Consequently $\prob_i(v) \geq p(\Delta,\ell)$ for all $i$, and (b) follows.  Note that (c) follows immediately from (b).  Similarly, (d) follows from (b) and Lemma \ref{lem:bigclique}.

To see that (e) holds, simply note that $R_i$ is always a maximal stable set in $H_i$.  Therefore if $w(v)<1$, then $\capacity_\nu(v) > 0$, thus $v\in H_i$ for every $i$, meaning that $R_i$ intersects $\tilde N(v)$ with probability 1.
\end{proof}

\subsection{Maximizing the expenditure}\label{sec:stopping}

Here we consider the best possible choice of $y$ in Lemma \ref{lem:initial}.  The optimal value of $y$ will be the largest possible such that the upper bound on $\rho_{\1-w}(G)$ is still achieved by some vertex in $G_\omega$.  If we increase $y$ beyond this point, we will find that $\rho_{\1-w}(G)$ is no longer guaranteed to drop as fast as $y$ increases.

In light of this goal, for $1\leq k \leq 3$ we let $\tilde y_k(\Delta)$ denote the maximum value of $y$ such that
\begin{equation}\label{eq:y1}
(1+\mu(\Delta)) y \leq \tfrac 12(\Delta-1-k) + \left(\tfrac 12 + \tfrac 12 k\mu(\Delta) \right) y.
\end{equation}
For $4\leq k \leq \Delta-2$ we let $\tilde y_k(\Delta)$ denote the maximum value of $y$ such that
\begin{equation}\label{eq:y2}
(1+\mu(\Delta)) y \leq \tfrac 12(\Delta-1-k) + \left(\tfrac 12 + \tfrac 12 k\mu_{\Delta+1-k}(\Delta) \right) y.
\end{equation}
Now let $\tilde y(\Delta)$ denote $\min \{\min_k \tilde y_k(\Delta), \omega, \frac{\omega-3}{1-3\mu(\Delta)}\}$ (the latter two bounds are for convenience of proof, and do not affect our results).  Our initial colouring phase culminates in the following consequence of Lemma \ref{lem:initial}.

\begin{theorem}\label{thm:initial}
For any $0\leq y\leq \tilde y(\Delta)$, there is a vertex weighting $w$ and fractional $y\wr w$-colouring of $G$ such that $\rho_{\1 - w}(G)\leq \Delta-(1+\mu(\Delta))y$.
\end{theorem}

\begin{proof}
Let $v$ be any vertex in $G$; it suffices to prove that $\rho_{\1-w}(v)\leq \Delta-(1+\mu(\Delta))y$.  We take the fractional $y \wr w$-colouring guaranteed by Lemma \ref{lem:initial}.

First suppose $v\in G_\omega$, and assume without loss of generality that $v\in B_1$.  We know that $w(B_1) = y$ by Lemma \ref{lem:initial}(a), and that for any $u\in \tilde N(v)\setminus B_1$, $w(u) \geq y\mu(\Delta)$ (by Lemma \ref{lem:initial}(b)).  Therefore $|\tilde N(v)|-w(\tilde N(v)) \leq \omega-y + 2(1-y\mu(\Delta)) = \Delta+1-y-2y\mu(\Delta)$.  We now claim that for any clique $C$ containing $v$, $|C|-w(C) \leq \omega-y$.  Clearly $w(B_1) = y$.  For $C$ not equal to $B_1$, Lemma \ref{lem:mce} tells us that $|C|\leq 3$.  Therefore $|C|-w(C) \leq  3-3y\mu(\Delta)$.  If $\omega-y < 3-3y\mu(\Delta)$, then $\omega-3 < y(1-3\mu(\Delta))$, contradicting the fact that $y \leq \tilde y(\Delta) \leq \frac{\omega-3}{1-3\mu(\Delta)}\}$.  Therefore $|C|-w(C) \leq \omega-y=\Delta-1-y$.  Thus
\begin{equation}
\rho_{\1-w}(v) \leq \tfrac 12(\Delta-1-y) + \tfrac 12(\Delta+1-y-2y\mu(\Delta)) = \Delta - (1+\mu(\Delta))y.
\end{equation}

Now suppose that $v$ is not in $V_\omega$, and let $C$ be a clique containing $v$ such that $|C|-w(C)$ is maximum.  Denote the size of $C$ by $k$.  By Lemma \ref{lem:initial}(e), we know that $w(\tilde N(v))\geq y$, so 
\begin{equation}
|\tilde N(v)|-w(\tilde N(v))\leq \Delta+1-y.
\end{equation}
Therefore to prove that $\rho_{\1-w}(v)\leq \Delta - (1+\mu(\Delta))y$, it is sufficient to prove that
\begin{equation}
k-w(C)\leq \Delta-1 -y - 2y\mu(\Delta),
\end{equation}
i.e.
\begin{equation}
(\mu(\Delta)+\tfrac 12)y\leq \tfrac 12(\Delta-1-k) +\tfrac 12w(C).
\end{equation}
By Lemma \ref{lem:initial}(c) we know that $w(C) \geq k\mu(\Delta)y$.  If $k\geq 4$, by Lemma \ref{lem:initial}(d) we know that $w(C)\geq k\mu_{\Delta+1-k}(\Delta)y$.  We also know that $y \leq \tilde y(\Delta) \leq \tilde y_k(\Delta)$, so if $k\leq 3$ then 
\begin{equation}
(1+\mu(\Delta)) y \leq \tfrac 12(\Delta-1-k) + \left(\tfrac 12 + \tfrac 12 k\mu(\Delta) \right) y,
\end{equation}
and if $k\geq 4$ then
\begin{equation}
(1+\mu(\Delta)) y \leq \tfrac 12(\Delta-1-k) + \left(\tfrac 12 + \tfrac 12 k\mu_{\Delta+1-k}(\Delta) \right) y.
\end{equation}
In either case,
\begin{equation}
(1+\mu(\Delta)) y \leq \tfrac 12(\Delta-1-k)+ \left(\tfrac 12y + \tfrac 12 w(C)\right),
\end{equation}
so
\begin{equation}
(\mu(\Delta)+\tfrac 12) y \leq \tfrac 12(\Delta-1-k)+ \tfrac 12 w(C),
\end{equation}
as desired.  Thus $\rho_{\1-w}(v)\leq \Delta - (1+\mu(\Delta))y$.
\end{proof}

Since equations \ref{eq:y1} and \ref{eq:y2} are linear, we can easily find the optimal values of $\tilde y_k(\Delta)$ by solving for
\begin{equation}
\tilde y_k(\Delta) =  \frac{\tfrac 12(\Delta-1-k)}{\tfrac 12+\mu(\Delta)-\tfrac 12k\mu(\Delta)}
\end{equation}
when $k\leq 3$ and for 
\begin{equation}
\tilde y_k(\Delta) =  \frac{\tfrac 12(\Delta-1-k)}{\tfrac 12+\mu(\Delta)-\tfrac 12k\mu_{\Delta+1-k}(\Delta)}
\end{equation}
when $\Delta-2 \geq k \geq 4$.  See \cite{sagekdeltafree} and Table \ref{tab:mu} for numerical values.

\section{Proving the main result}\label{sec:main}

We now have enough results in hand to prove the main result easily.

\begin{theorem}\label{thm:main}
For $\Delta \geq 6$, let $G$ be a graph with maximum degree $\Delta$ and clique number at most $\Delta-1$.  Then $G$ has fractional chromatic number at most $\Delta - \min\{\frac 12, \tilde y(\Delta)\mu(\Delta)\}$.
\end{theorem}

\begin{proof}
Let $G$ be a minimum counterexample; Theorem \ref{thm:mcdiarmid} tells us that $G$ has maximum degree $\Delta$ and clique number $\omega=\Delta-1$.  Lemma \ref{lem:mce} tells us that all $\omega$-cliques of $G$ are disjoint, and that no vertex $v$ has two neighbours in an $\omega$-clique not containing $v$.

We may therefore set $y= \tilde y(\Delta)$ and apply Theorem \ref{thm:initial}.  This gives us a vertex weighting $w$ and fractional $y\wr w$-colouring of $G$ such that $\rho_{\1 - w}(G)\leq \Delta-(1+\mu(\Delta))y$.  By Theorem \ref{thm:weighted}, $\chi_f^{\1-w} \leq \rho_{\1 - w}(G)\leq \Delta-(1+\mu(\Delta))y$.  That is, there is a fractional $(\Delta-(1+\mu(\Delta))y)\wr (\1 - w)$-colouring of $G$.  Combining this colouring with the initial fractional $y\wr w$-colouring gives us a fractional $(\Delta-\tilde y(\Delta)\mu(\Delta))$-colouring, which tells us that $\chi_f(G)\leq \Delta-\tilde y(\Delta)\mu(\Delta)$.
\end{proof}
For all values of $\Delta$ we have investigated, $\tilde y(\Delta)\mu(\Delta) < \tfrac 15$.  We believe that this is always the case.

\section{The structural reduction}\label{sec:mce}

In this section we prove Lemma \ref{lem:mce}, which tells us that we need only consider graphs whose maximum cliques behave nicely.  First observe that every proper induced subgraph of $G$ is fractionally $(\Delta-\epsilon)$-colourable, since deleting vertices from a graph with $\Delta=5$ cannot create a copy of $C_5\boxtimes K_2$.  We prove the lemma in two parts:

\begin{lemma}\label{lem:disjoint}
Part (i) of Lemma \ref{lem:mce} holds.
\end{lemma}

\begin{lemma}\label{lem:mce2}
Part (ii) of Lemma \ref{lem:mce} holds.
\end{lemma}

\subsection{Part (i)}

We actually split the proof of Lemma \ref{lem:disjoint} into three parts.  Suppose $C$ and $C'$ are two intersecting $\omega$-cliques.  Since $\omega=\Delta-1$, we can immediately observe that $|C\cap C'|\geq \omega-2$. Therefore Lemma \ref{lem:disjoint} follows as an easy corollary of the next three Lemmas \ref{lem:3cliques}, \ref{lem:2cliques}, \ref{lem:2cliques2}.  Throughout this section we will make implicit use of the fact that every vertex in $G$ has at least $\Delta-1$ neighbours, as is trivially implied by the minimality of $G$.  Furthermore note that whenever we reduce $G$ to a graph $G'$, no component of which is $5$-regular, no component of $G'$ can be isomorphic to $C_5\boxtimes K_2$.

\begin{lemma}\label{lem:3cliques}
$G$ does not contain three $\omega$-cliques mutually intersecting in $\omega-1$ vertices.
\end{lemma}

\begin{proof}
Suppose that $G$ contains an $(\omega-1)$-clique $X$ and vertices $x_1,x_2,x_3$ each of which is complete to $X$.  Because there is no $(\omega+1)$-clique,  $\{x_1,x_2,x_3\}$ is a stable set. 
Let $G' = G\setminus (X \cup \{x_1,x_2,x_3\})$; as previously observed, since $G'$ is a proper induced subgraph of $G$, there is  a fractional ${(\Delta-\epsilon)}$-colouring $\kappa$ of $G'$. We extend $\kappa$ to a fractional $(\Delta-\epsilon)$-colouring of $G$ to obtain a contradiction, beginning by colouring $\{x_1,x_2,x_3\}$ using weight at most $2-\epsilon$.

First suppose $\Delta=5$, so $\epsilon \leq \frac 13$.  Since each $x_i$ has at most two neighbours in $G'$, we have $|\alpha(x_i)| \geq \Delta-\epsilon-2$.  Note that $|\alpha(x_i) \cup \alpha(x_j)| \leq \Delta-\epsilon$, so for any $\{i,j\}\subseteq \{1,2,3\}$ we have $|\alpha(x_i)\cap \alpha(x_j)| \geq \Delta-\epsilon-4 \geq 1-\epsilon\geq \frac 23$.  We extend $\kappa$ to $\{x_1,x_2,x_3\}$ such that
\begin{itemize}
\item $|\kappa[x_1]\cap \kappa[x_2]| \geq \frac 23$, and 
\item There exist disjoint subsets $s_1$ and $s_2$ of $\kappa[x_3]$, each of size $\frac 13$, such that $s_1\subset \kappa[x_1]$ and $s_2\subset \kappa[x_2]$.
\end{itemize}
To do this, we first give $x_1$ and $x_2$ weight $\frac 23$ of colour in common, then give $x_1$ and $x_3$ weight $\frac 13$ of colour each such that all the colour on $x_3$ is in $\kappa[x_1]$, then give $x_2$ and $x_3$ weight $\frac 13$ of colour each such that all the new colour on $x_3$ is in $\kappa[x_2]$.  Finally we complete the colouring of $x_3$ arbitrarily.  Confirming that this is possible is straightforward given the pairwise intersections of $\alpha(x_i)$.  Furthermore since $|\kappa[\{x_1,x_2\}]|\leq \frac 43$ and at least $\frac 23$ of the colour in $\kappa[x_3]$ is in $\kappa[\{x_1,x_2\}]$, we use weight at most $2-\epsilon$ on $\{x_1,x_2,x_3\}$.

Now suppose $\Delta\geq 6$, so $\epsilon \leq \frac 12$.  Our approach is the same as before, except now for any $\{i,j\}\subseteq \{1,2,3\}$ we have $|\alpha(x_i)\cap \alpha(x_j)| \geq \Delta-\epsilon-4 \geq \frac 32$.  Thus we can proceed by giving $x_1$ and $x_2$ weight $\frac 12$ of colour in common, then assign $s_1$ and $s_2$ as before, but with size $\frac 12$ each.  Again we use weight at most $2-\epsilon$ on $\{x_1,x_2,x_3\}$.

We now have $\{v\in V(G) : \left|\kappa(v)\right|<1 \} =V(X)$.  For every $v\in V(X)$, we have $\left|\alpha(v)\right| \geq \Delta-\epsilon-(2-\epsilon)=\omega-1 = |V(X)|$.  We may therefore apply Lemma \ref{lem:hall} and extend $\kappa$ to a fractional $(\Delta-\epsilon)$-colouring of $G$.
\end{proof}

\begin{lemma}\label{lem:2cliques}
$G$ does not contain two $\omega$-cliques intersecting in $\omega-1$ vertices.
\end{lemma}

\begin{proof}
Suppose $C$ and $C'$ are two $\omega$-cliques intersecting in $\omega-1$ vertices.
Let $v_1,\ldots,v_{\omega-1}$ be the vertices in $C\cap C'$, let $x$ be the vertex in $C\setminus C'$, and let $y$ be the vertex in $C'\setminus C$, noting that $x$ and $y$ are nonadjacent.  For $1\leq i \leq \omega-1$, if $v_i$ has a neighbour outside $C\cup C'$ call it $u_i$.  

\begin{claim}
There exists a fractional $(\Delta-\epsilon)$-colouring $\kappa$ of $G\setminus (C\cap C')$ satisfying the following:

\begin{enumerate}\label{clm:2cliques}
\item[(1)] If $\Delta = 5$, then $|\kappa[\{x,y\}]| \leq 1+\epsilon$.
\item[(2)] If $\Delta \geq 6$, then $|\kappa[\{x,y\}]| = 1$.
\item[(3)] $|\bigcap_{i\leq \omega-1} \kappa[u_i]| \leq \epsilon$.
\end{enumerate}
\end{claim}

We first show how the claim implies the lemma. For each $v_i\in C\cap C'$, $|\alpha(v_i)| \geq \Delta-\epsilon - |\kappa[\{x,y,u_i\}]| \geq \Delta-\epsilon-1-|\kappa[\{x,y\}]| \geq \omega-2$. Thus to apply Lemma \ref{lem:hall} and extend $\kappa$ to $G$ it is enough to show that $|\bigcup_{i\leq \omega-1}\alpha(v_i)| \geq \omega -1$. Indeed, the set of colours available to at least some of the vertices in $C\cap C'$ are those which are not forbidden to all of them: If $\Delta \geq 6$, then $$\left|\bigcup_{i\leq \omega-1}\alpha(v_i)\right| \geq \Delta-\epsilon - \left|\kappa[\{x,y\}]\right| - \left|\bigcap_{i\leq \omega-1} \kappa[u_i]\right| \geq \omega-2\epsilon \geq \omega-1 $$ and if $\Delta = 5$, $$\left|\bigcup_{i\leq \omega-1}\alpha(v_i)\right| \geq \omega-3\epsilon \geq \omega-1. $$

Lemma \ref{lem:hall} then guarantees a fractional $(\Delta-\epsilon)$-colouring of $G$, a contradiction. 

\begin{proof}[Proof of Claim \ref{clm:2cliques}]

There are two cases.  Note that by Lemma \ref{lem:3cliques}, if $u_i$ exists for each $i$ then $|\{u_i: 1 \leq i \leq \omega-1\}|\geq 2$.

\vspace{.5em}\noindent{\bf Case 1}: $2\leq |\{u_i: 1 \leq i \leq \omega-1\}| < \omega-1$ and $u_i$ exists for each $i$.

\noindent Without loss of generality suppose that $u_1=u_2$ and consider $G'= G\setminus(C\cup C'\cup \{u_1\})$. Again, since $G'$ is a proper induced subgraph of $G$, there exists a fractional $(\Delta-\epsilon)$-colouring $\kappa$ of $G'$. We extend $\kappa$ to a fractional colouring of $G\setminus(C\cap C')$, first colouring $x$ and $y$, then $u_1$.

Each of $x$ and $y$ has at most two neighbours in $G'$ so we have $|\alpha(x)|, |\alpha(y)| \geq \Delta-\epsilon-2$. Since $|\alpha(x)\cup\alpha(y)| \leq \Delta-\epsilon$ it follows that $|\alpha(x) \cap \alpha(y)| \geq \Delta-\epsilon-4 \geq 1$ when $\Delta \geq 6$, and $|\alpha(x) \cap \alpha(y)| \geq 1-\epsilon$ when $\Delta=5$. We extend $\kappa$ in the obvious way so that if $\Delta\geq 6$ then $\kappa[x]=\kappa[y]$, and if $\Delta=5$ then $|\kappa[x]\cap \kappa[y]| \geq 1-\epsilon$, satisfying (1) and (2). It remains to colour $u_1$. Note that $u_1$ has degree at most $\omega-1$ in $G\setminus (C\cap C')$ so $|\alpha(u_1)| \geq 2-\epsilon$. Because $|\bigcap_{3\leq i\leq \omega-1} \kappa[u_i]| \leq 1$, we can choose $\kappa[u_1]$ from $\alpha(u_1)$ in such a way that $|\kappa[u_1]\cap \bigcap_{3\leq i\leq \omega-1} \kappa[u_i]| \leq \epsilon$, satisfying (3).

\vspace{.5em}\noindent{\bf Case 2:} $|\{u_i: 1 \leq i \leq \omega-1\}| = \omega-1$ or $u_i$ does not exist for some $i$.

\noindent If there exists an edge $u_iu_j$ in $G$ for some $i \neq j$, then let $G' = G\setminus (C\cup C')$. Otherwise choose $i\neq j$ such that adding the edge $u_iu_j$ to $G\setminus (C\cup C')$ yields a graph with $\omega <\Delta$ and let $G'= G\setminus (C\cup C') \cup u_iu_j$. To see that such $i$ and $j$ exist, consider $u_1,u_2$ and $u_3$ and suppose that each pair of these has an $(\omega-1)$-clique in the common neighbourhood. Because $\Delta=\omega+1$ there must be a vertex contained in each of these three cliques, but Lemma \ref{lem:3cliques} forbids the existence of three pairwise intersecting $\omega$-cliques.

By the minimality of $G$, there exists a fractional $(\Delta-\epsilon)$-colouring $\kappa$ of $G'$. We need to extend $\kappa$ to $x$ and $y$. Because each of $x$ and $y$ has at most two neighbours in $G'$ we have $|\alpha(x)|, |\alpha(y)| \geq \Delta-\epsilon-2$. It follows that $|\alpha(x)\cap \alpha(y)| \geq 1$ if $\Delta\geq 6$ and $|\alpha(x)\cap \alpha(y)| \geq 1-\epsilon$ if $\Delta=5$ so we can extend $\kappa$ in the obvious way to satisfy (1) and (2). Requirement (3) is guaranteed by the existence of the edge $u_iu_j$.  This proves the claim.
\end{proof}
As we have shown, the claim implies the lemma.
\end{proof}

\begin{lemma}\label{lem:2cliques2}
$G$ does not contain two $\omega$-cliques intersecting in $\omega-2$ vertices. 
\end{lemma}

\begin{proof}
Suppose $C$ and $C'$ are two $\omega$-cliques intersecting in $\omega-2$ vertices. Let $x,x'$ be the vertices in $C\setminus C'$ and let $y,y'$ be those in $C'\setminus C$.  Suppose that $x$ is adjacent to $y$.  Then $C$ and $(C\setminus \{x'\})\cup \{y\})$ are two $\omega$-cliques intersecting in $\omega-1$ vertices, contradicting Lemma \ref{lem:2cliques}.  By symmetry we may therefore assume there is no edge between $\{x,x'\}$ and $\{y,y'\}$.  The case $\Delta=5$ gives us the most difficulty by far, so we deal with it separately.

\vspace{.5em}\noindent{\bf Case 1:} $\Delta\geq 6$.

\noindent We construct the graph $G'$ from $G$ by identifying $x,y$ and $x',y'$ into two new vertices $z$ and $z'$, respectively, and deleting $C\cap C'$. Clearly $\Delta(G') \leq \Delta(G)$.  If $G'$ contains a $\Delta$-clique, then since $z$ and $z'$ have degree at most 5, we have $\Delta = 6$, and furthermore the $\Delta$-clique must contain both $z$ and $z'$.  Thus there is a set of four vertices $C''$ forming a $6$-clique with $z$ and $z'$.  This means there must be eight edges between $\{x,x',y,y'\}$ and $C''$ in $G$.

If any vertex in $C''$ has a neighbour outside of $\{x,x',y,y'\}$ then $C''$ is a clique cutset in $G$, contradicting the fact that every proper induced subgraph of $G$ is fractionally $(\Delta-\epsilon)$-colourable. Thus $V(G)=V(C)\cup V(C')\cup V(C'')$. Further, $(N(x)\cup N(y))\cap V(C'') = V(C'')$ and $(N(x')\cup N(y'))\cap V(C'') = V(C'')$. If $x$ and $x'$ have the same two neighbours in $C''$ then $G$ is the graph $(C_5\boxtimes K_3)-4v$ shown in Figure \ref{fig:upperbounds}, contradicting the assumption that $\chi_f(G)>\Delta-\frac 12$. Thus $x$ and $y'$ have a common neighbour in $C''$. We may safely switch the roles of $y$ and $y'$ in this case to ensure that $\omega(G')\leq \omega(G)$.

It now follows from the minimality of $G$ that there exists a fractional $(\Delta-\epsilon)$-colouring $\kappa$ of $G'$. By unidentifying $x,y$ and $x',y'$, we may think of $\kappa$ as a fractional colouring of $G\setminus(C\cap C')$ where $\kappa[x]=\kappa[y]$ and $\kappa[x']=\kappa[y']$.  We now  extend $\kappa$ to a $(\Delta-\epsilon)$-colouring of $G$. We have $\{v\in V(G) : |\kappa(v)|<1 \} =V(C\cap C')$. Further, for each $v\in V(C\cap C')$, $|\alpha(v)| \geq \Delta-\epsilon-2 \geq \omega-2$.  Thus applying Lemma \ref{lem:hall} gives the extension of $\kappa$ to $G$, a contradiction.

\vspace{.5em}\noindent{\bf Case 2:} $\Delta=5$.

\noindent We construct $G'$ as in the previous case.  If $G'$ has a fractional $(\Delta-\epsilon)$-colouring, we reach a contradiction as before.  Otherwise, it must be the case that $G'$ contains a $\Delta$-clique or  $C_5\boxtimes K_2$.  To deal with these cases we prove four claims.

{\bf Our first claim} is that no vertex in $G\setminus(C\cup C')$ has a neighbour in both $\{x,x'\}$ and $\{y,y'\}$.  To prove this, assume for a contradiction that $x$ and $y$ have a common neighbour $w\notin C\cup  C'$.  Let $G'' = G\setminus(C\cup C')$. By the minimality of $G$ there exists a fractional $(\Delta-\epsilon)$-colouring $\kappa$ of $G''$ that we now extend to a fractional colouring of $G$. We do so in two steps, first colouring $\{x,y,x',y'\}$. 

Since $x$ and $y$ have a common neighbour plus at most one other coloured neighbour each, we have $|\alpha(x)\cap\alpha(y)| \geq \Delta-\epsilon-3$. On the other hand, each of $x'$ and $y'$ has at most two coloured neighbours, so $|\kappa[N(x')\cup N(y')]| \leq 4$.  We choose $\kappa[x]=\kappa[y]$ from $\alpha(x)\cap\alpha(y)$ maximizing its intersection with $\kappa[N(x')\cup N(y')]$, so that after colouring $x$ and $y$ we still have $|\kappa[N(x')\cup N(y')]| \leq 4$.   This means that $|\alpha(x')\cap\alpha(y')| \geq 1-\epsilon$ so we may choose colours for $x'$ and $y'$ so that $|\kappa[x']\cap\kappa[y']| \geq 1-\epsilon$. This ensures that $|\kappa[\{x,y,x',y'\}]|\leq 2+\epsilon$.

It remains to extend the colouring to the vertices in $C\cap C'$.  For each vertex $v \in V(C\cap C')$, $|\alpha(v)| \geq \Delta-\epsilon-(2+\epsilon) \geq \omega-2$. Applying Lemma \ref{lem:hall}, we find a fractional $(\Delta-\epsilon)$-colouring of $G$, a contradiction.  This proves the first claim, so we may henceforth assume no vertex in $G\setminus(C\cup C')$ has a neighbour in both $\{x,x'\}$ and $\{y,y'\}$.

{\bf Our second claim} is that $G$ does not contain an edge cut of size at most two.  For if it does, we can take a fractional $(\Delta-\epsilon)$-colouring of either side of this cut.  The edges of the cut have colour weight at most four on their endpoints, and since $\Delta-\epsilon > 2\cdot 2$, we can safely merge the $(\Delta-\epsilon)$-colouring of either side of the cut into a fractional $(\Delta-\epsilon)$-colouring of $G$, a contradiction.  This proves the second claim.

{\bf Our third claim} is that $G'$ does not contain a $\Delta$-clique.  Suppose it does; we now investigate the structure of $G$.  In $G\setminus(C\cup C')$ there is an $\omega-1$ clique $C''$, each vertex of which is complete (in $G$) to either $\{x,x'\}$ or $\{y,y'\}$, since no vertex has neighbours in both $\{x,x'\}$ and $\{y,y'\}$ (by the first claim).  Since $|C''|=3$, we may assume that $x$ and $x'$ have two common neighbours $w_1$ and $w_2$ in $C''$, and $y$ and $y'$ have a common neighbour $w_3$ in $C''\setminus\{w_1,w_2\}$.  Call the neighbours of $y$ and $y'$ in $G\setminus(C'\cup C'')$ $v$ and $v'$ respectively, if these vertices exist. We assume $v$ and $v'$ exist, as adding them as pendant vertices does not affect the proof adversely.  Let $G''$ be the graph obtained from $G\setminus(C\cup C'\cup C'')$ by adding the edge $vv'$ if possible ($v$ and $v'$ may not be two distinct vertices, or may already be adjacent).  This construction does not create a $\Delta$-clique in $G''$ since no pair of cliques in $G$ intersects in $\omega-1$ vertices by Lemma \ref{lem:2cliques}.  Bearing in mind that $\Delta=5$, $G''$ cannot contain a copy of $C_5\boxtimes K_2$, since the existence of $(C_5\boxtimes K_2)-e$ in $G$ would violate the second claim.  Therefore the minimality of $G$ guarantees that $G''$ has a fractional $(\Delta-\epsilon)$-colouring $\kappa$.  We extend in two cases based on whether or not $|\{v,v'\}|=2$.

\begin{figure}
\begin{center}
\begin{tabular}{cc}
\includegraphics[scale=.7]{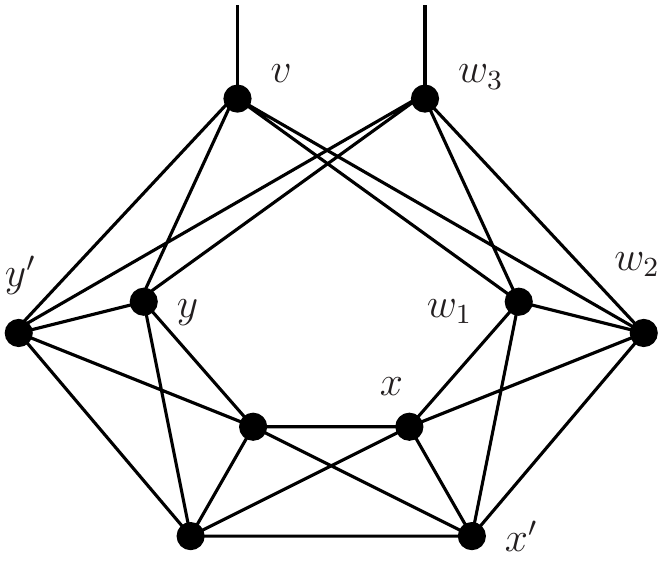}
& \hspace{.5in}
\includegraphics[scale=.7]{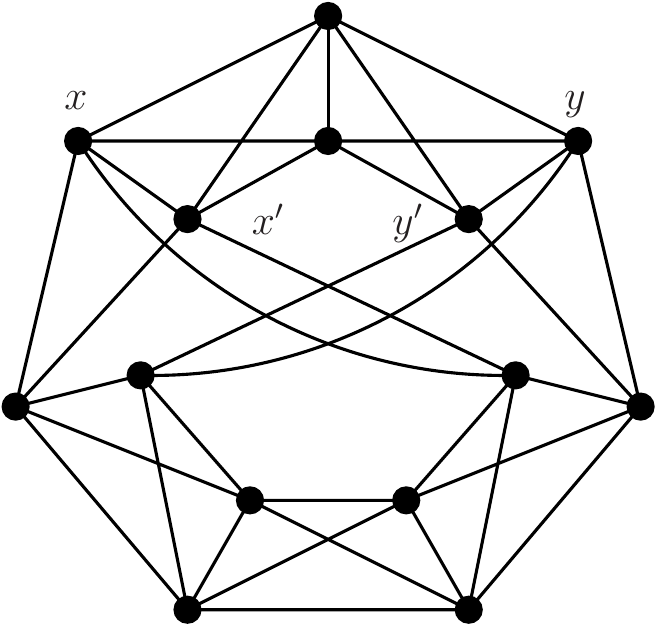}
\end{tabular}
\caption{Left: If $G$ contains $(C_5\boxtimes K_2)-e$, we can easily reduce. Right: Reducing on the six top vertices renders $G'$ isomorphic to $C_5\boxtimes K_2$.}\label{fig:separable}
\end{center}
\end{figure}

Note that if $|\{v,v'\}| = 1$, we may assume one of $w_1,w_2$ is nonadjacent to $v$, say $w_1$ is nonadjacent to $v$, otherwise $G$ contains a copy of $(C_5\boxtimes K_2)-e$, violating the second claim (see Figure \ref{fig:separable} (left)).  Now assume $|\{v,v'\}|\leq 1$.  We recolour $v$ (if it exists) such that $|\kappa[v] \cap (\alpha(w_1)\cup\alpha(w_2))| \geq 1-\epsilon$.  This is possible because $|\alpha(v)|\geq 2-\epsilon$ and $|\alpha(w_1)|\geq 4-\epsilon$, so the intersection of these two sets is at least $(6-2\epsilon)-(5-\epsilon) = 1-\epsilon$.  Now we may easily extend $\kappa$ by colouring $w_1$ such that $|\kappa[v] \cap \kappa[w_1]| \geq 1-\epsilon$.  Next we extend $\kappa$ by colouring $w_2$ and $w_3$, which we can do greedily since each of these vertices has at most three neighbours in $G\setminus(C\cup C')$.  Now it remains to colour $C\cup C'$.  Since $|\kappa[\{v,w_1,w_2,w_3\}]| \leq 4-( 1-\epsilon)$, there is weight $\frac 43$ of colour we can use on both $\{x,x'\}$ and $\{y,y'\}$.  Since each vertex in $\{x,x',y,y'\}$ has only three neighbours in $G\setminus (C\cap C')$, we can extend $\kappa$ to a colouring of $G\setminus (C\cap C')$ such that $|\kappa[\{x,x'\}] \cap \kappa[\{y,y'\}]| \geq \frac 43$.  After doing this we can easily extend $\kappa$ to a fractional $(5-\epsilon)$-colouring of $G$ by applying Lemma \ref{lem:hall}, a contradiction.

Now we handle the case $|\{v,v'\}|=2$, starting with a fractional $(5-\epsilon)$-colouring of $G''$ which we take as a partial coloring of $G$.  We begin by extending $\kappa$ by colouring $w_3$ such that $\kappa[w_3]\subset \kappa[\{v,v'\}]$, which is possible because $\kappa[v]$ and $\kappa[v']$ are disjoint (and $w_3$ is adjacent to at most one of $v$ and $v'$, since it is adjacent to $w_1$, $w_2$, $y$ and $y'$).  We now extend $\kappa$ by colouring $w_1$ and $w_2$ in any way, which we can do greedily.  At this point, we have $|\alpha(y)|\geq \frac 83$, $|\alpha(y')|\geq \frac 83$, and $|\alpha(y)\cup \alpha(y')|\geq \frac {11}3$.  Therefore $|\alpha(y)\setminus \kappa[\{w_1,w_2\}]|\geq \frac 23$, $|\alpha(y')\setminus \kappa[\{w_1,w_2\}]|\geq \frac 23$, and $|(\alpha(y)\cup \alpha(y'))\setminus \kappa[\{w_1,w_2\}]|\geq \frac 53$.  We may therefore give $y$ weight $\frac 23$ of colour not in $\kappa[\{w_1,w_2\}]$, and give $y'$ weight $\frac 23$ of colour not in $\kappa[\{w_1,w_2\}]$, then finish colouring $y$ and $y'$ greedily, since each has at most three neighbours in $G\setminus C$.  It follows that $|\kappa[\{w_1,w_2\}] \cap \kappa[\{y,y'\}]|\leq \frac 23$, so we can extend $\kappa$ by colouring $\{x,x'\}$ such that $|\kappa[\{w_1,w_2\}] \cap \kappa[\{x,x'\}]|\geq \frac 43$.  We can now extend $\kappa$ to a fractional $(\Delta-\epsilon)$-colouring of $G$ by applying Lemma \ref{lem:hall} as in the previous case.  This contradiction proves the third claim.

{\bf Our fourth claim}, which is sufficient to complete the proof, is that $G'$ does not contain $C_5\boxtimes K_2$.  If it does, there must be four vertices $w$, $w'$, $v$, and $v'$ such that in $G'$, $\{w,w',z,z'\}$ and $\{v,v',z,z'\}$ are cliques.  Each of $w$, $w'$, $v$, and $v'$ therefore has two neighbours in $\{x,x',y,y'\}$.  By the first claim, there are two cases, by symmetry: $w$ and $w'$ are adjacent to both $x$ and $x'$, or $w$ and $v$ are adjacent to both $x$ and $x'$.  In the first case, the component of $G$ containing $C$ is isomorphic to $C_7\boxtimes K_2$, a contradiction since $\chi_f(C_7\boxtimes K_2)=\frac{14}3$.  In the second case, the component of $G$ containing $G$ is isomorphic to the graph shown in Figure \ref{fig:separable} (right).  Observe that the outer seven vertices induce $C_7$, as do the inner seven vertices.  Therefore $\chi_f(G) \leq 2\chi_f(C_7) = 5-\tfrac 13$, a contradiction.  This completes the proof of the lemma.
\end{proof}

\subsection{Part (ii)}

Our approach to proving Lemma \ref{lem:mce2} involves reducing $G$ to a smaller graph $G'$.  Either $G'$ is fractionally $(\Delta-\epsilon)$-colourable by minimality, in which case we finish easily, or $G'$ contains a $K_\Delta$ or $C_5\boxtimes K_2$ (when $\Delta=5$), in which case we proceed on a case-by-case basis.

To simplify things, we first need to prove a couple of lemmas that exclude induced subgraphs of $G$.

\begin{figure}
\begin{center}
\includegraphics[scale=.7]{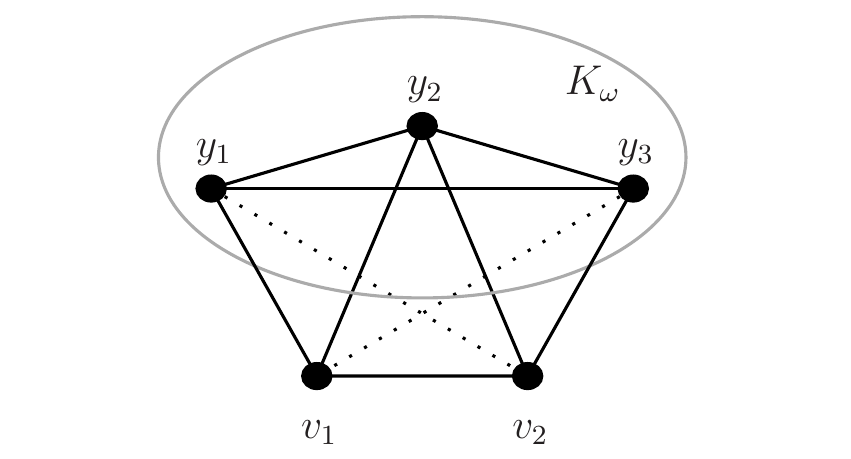}
\end{center}\caption{A bump.}\label{fig:bump}
\end{figure}

\begin{dfn}
Suppose we have a set $Y= \{y_1,y_2,y_3\}$ of vertices in a maximum clique $C$, and two adjacent vertices $v_1$ and $v_2$ such that $N(v_1)\cap Y = \{y_1,y_2\}$ and $N(v_2)\cap Y = \{y_2,y_3\}$.  Then we say that the set $X = C\cup \{v_1,v_2\}$ is a {\em bump} (see Figure \ref{fig:bump}).
\end{dfn}

\begin{lemma}\label{lem:bump}
$G$ does not contain a bump.
\end{lemma}

\begin{proof}
Suppose to the contrary that $G$ contains a bump $X$.  To reach a contradiction we take a fractional $(\Delta-\epsilon)$-colouring $\kappa$ of $G' = G \setminus X$ and extend it to a $(\Delta-\epsilon)$-colouring of $G$ as follows.

First, extend $\kappa$ by colouring $v_1$ and $y_3$ with the same set of colours.  This is possible because $v_1$ has at most $\Delta-3$ neighbours in $G'$, and $y_3$ has at most one neighbour in $G'$, so $|\alpha(v_1) \cap \alpha(y_3)| \geq \Delta-\epsilon - (\Delta-3) -1 > 1$.

Next we extend $\kappa$ by giving $v_2$ and $y_1$ common colour of total weight $\frac 12$, and leaving them only partially coloured.  This is possible because at this point, $v_2$ has at most $\Delta-1$ coloured neighbours, and $y_1$ has at most 3 coloured neighbours, but both are adjacent to $v_1$ and $y_3$. Therefore $|\kappa[N(y_1)\cup N(v_2)]|\leq \Delta-2 + 1 = \Delta-1$, and so $|\alpha(y_1)\cap \alpha(v_2)|\geq 1-\epsilon\geq \tfrac 12$.

At this point observe that $|\kappa[Y]| = \frac 32\leq (\Delta-\epsilon)-2 -(\Delta-4)$, so we may now greedily extend $\kappa$ by colouring the $\Delta-4$ vertices in $C\setminus Y$, since each of these has at most two coloured neighbours in $G'$.  All that remains is to complete the colouring of $v_2$, $y_1$, and $y_2$.  First we finish colouring $y_1$; we can do this greedily because at this point $|\kappa[N(y_1)]| \leq \Delta-2$, since $y_2$ is uncoloured and $\kappa[v_1]=\kappa[y_3]$.  Next we greedily finish colouring $v_2$, which again we can do because at this point $|\kappa[N(v_2)]| \leq \Delta-2$, since $y_2$ is uncoloured and $\kappa[v_1]=\kappa[y_3]$.

Finally we must extend to $y_2$, which we can do greedily: since $\kappa[v_1]=\kappa[y_3]$ and $|\kappa[v_2]\cap \kappa[y_1]|\geq \frac 12$, $|\kappa[N(y_2)]| \leq \Delta-\frac 32$, so $|\alpha(y_2)|\geq \frac 32-\epsilon \geq 1$.  Thus $G$ is fractionally $(\Delta-\epsilon)$-colourable, a contradiction.
\end{proof}

We already know, thanks to Lemma \ref{lem:disjoint}, that $K_{\Delta}$ minus an edge cannot appear in $G$.  But given restrictions on $\Delta$, we can forbid other subgraphs arising as $K_{\Delta}$ minus a small number of edges.  We use variations of the approach for bumps: we extend a partial fractional colouring of the graph by leaving a set of vertices to the end, then finishing greedily, having already given their neighbourhoods lots of repeated colour.

\begin{figure}
\begin{center}
\includegraphics[scale=.45]{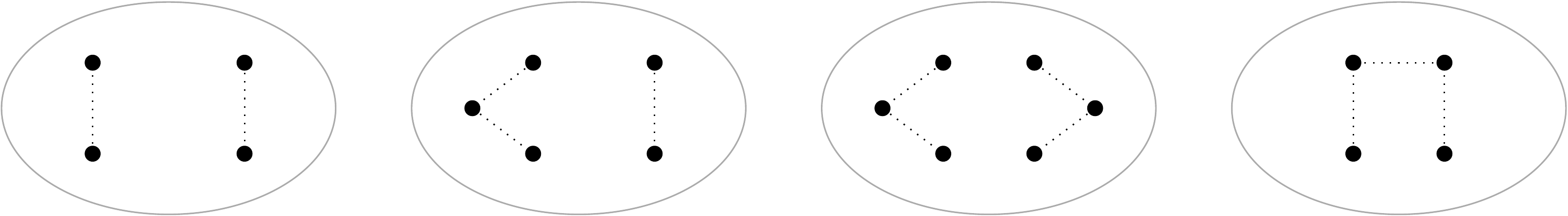}
\end{center}\caption{Configurations of edges missing from a $K_\Delta$ that are forbidden for, respectively, $\Delta\geq 5$ (Lemma \ref{lem:reduce1}), $\Delta\geq 6$ (Lemma \ref{lem:reduce2}), $\Delta \geq 7$ (Lemma \ref{lem:reduce3}), and $\Delta \geq 7$ (Lemma \ref{lem:reduce4}).}\label{fig:nearcliques}
\end{figure}

\begin{lemma}\label{lem:reduce1}
$G$ cannot contain $K_\Delta$ minus a matching of size two.
\end{lemma}

\begin{proof}
Suppose to the contrary that $G$ contains a subgraph $X$ on $\Delta$ vertices, with vertices $v_1,v_2,v_3,v_4\in V(X)$ such that the non-edges of $G[X]$ are exactly $\{v_1v_2,v_3v_4\}$.  We first consider the case where $\Delta\geq 6$.  We begin with a fractional $(\Delta-\epsilon)$-colouring $\kappa$ of $G'=G\setminus X$ and extend it to a $(\Delta-\epsilon)$-colouring of $G$ as follows.

First, we extend $\kappa$ by colouring $v_1$ and $v_2$ with the same set of colours. Each of $v_1,v_2$ has at most two coloured neighbours in $G'$, and so $|\alpha(v_1)\cap\alpha(v_2)| \geq (\Delta-\epsilon)-4 \geq 1$. Thus it is possible to choose $\kappa[v_1]=\kappa[v_2]$.

Next, we extend $\kappa$ by colouring $v_3$ and $v_4$ in such a way that $\kappa[v_3]\cap\kappa[v_4]\geq \tfrac 12$. Each of $v_3,v_4$ has at most two coloured neighbours in $G'$ as well as neighbours $v_1$ and $v_2$ which have the same set of colours, and so $|\alpha(v_3)\cap \alpha(v_4)| \geq (\Delta-\epsilon)-5 \geq 1-\epsilon \geq \epsilon \geq \tfrac 12$. Thus we may choose $\kappa[v_3]$ and $\kappa[v_4]$ as claimed. We now have $|\kappa[v_1,v_2,v_3,v_4]|\leq \tfrac 52$.

It remains to colour the $\Delta-4$ vertices in $X \setminus \{v_1,v_2,v_3,v_4\}$.  We can do this easily because for each such vertex, the total weight of colours appearing twice in its neighbourhood is at least $1+\epsilon$.  Therefore as we colour greedily, the weight on the closed neighbourhood will never exceed $\Delta-\epsilon$.  Thus $G$ is fractionally $(\Delta-\epsilon)$-colourable, a contradiction.

Now we consider the case where $\Delta=5$.  Let $u$ denote the neighbour of $v_5$ outside $X$; if $u$ does not exist, we can add a pendant vertex to $v_5$ and call it $u$, for the sake of our argument.  We begin with a fractional $(\Delta-\epsilon)$-colouring $\kappa$ of $G'=G\setminus X$ and extend it to a $(\Delta-\epsilon)$-colouring of $G$ as follows, considering three subcases based on $f = |\kappa[N(v_3)]\cap \kappa[N(v_4)]|$.  

If $f < \frac 13$, we give $v_1$ and $v_2$ common colour of weight $\frac 23$, leaving them only partially coloured.  We then put weight $\frac 23 -|\kappa[u]\cap \kappa[v_1]|$ of colour from $\kappa[u]\setminus \kappa[v_1]$ onto $\{v_3, v_4\}$ (putting none on both), which is possible because there is at least $\frac 23$ colour in $\kappa[u]\cap (\alpha(v_3)\cup \alpha(v_4))$.  We now extend $\kappa$ to completely colour $v_1$ and $v_2$, which is possible because at this point $|\kappa[\{v_3,v_4\}]| \leq \frac 23$.  Next we extend $\kappa$ to completely colour $v_3$ and $v_4$, which is possible because at this point $v_5$ is uncoloured and $|\kappa[v_1]\cap \kappa[v_2]|\geq \frac 23$.  Finally we extend the colouring to include $v_5$, which is possible because $|\kappa[v_1]\cap \kappa[v_2]|\geq \frac 23$ and $|\kappa[u]\cap \kappa[\{v_1,v_2,v_3,v_4\}]|\geq \frac 23$.  So we may assume $f\geq \frac 13$.

If $f < \frac 23$, we give $v_1$ and $v_2$ common colour of weight $\frac 23$, leaving them only partially coloured.  We then give $v_3$ and $v_4$ common colour of weight $\frac 13$, so at this point the total colour appearing on $N(v_3)\cup N(v_4)$ is at most $4-\frac 13+\frac 23 \leq 5-\epsilon-\frac 13$ (because $f\geq \frac 13$).  We then give $\{v_3,v_4\}$ enough colour from $\kappa[u]$ so that $|\kappa[u]\cap \kappa[\{v_1,v_2,v_3,v_4\}]|\geq \frac 13$; this is possible because $f < \frac 23$, and so $(\alpha(v_3)\cup\alpha(v_4)\cup\kappa[{v_1,v_2}])\geq 4$.  We may extend to finish colouring $v_1,v_2,v_3,v_4$ greedily, since $v_5$ is uncoloured and both $\kappa[v_1]\cap \kappa[v_2]$ and $\kappa[v_3]\cap \kappa[v_4]$ have size at least $\frac 13$.  Finally we can extend the colouring to $v_5$, since the weight of colours appearing at least twice on $N(v_5)$ is at least $\frac 43\geq 1+\epsilon$.  So we may assume $f\geq \frac 23$.

This final case is easiest: we give $v_1$ and $v_2$ common colour of weight $\frac 23$, then give $v_3$ and $v_4$ common colour of weight $\frac 23$, then extend to completely colour $\{v_1,v_2,v_3,v_4\}$ greedily, then extend to $v_5$ greedily.  The details are as in the previous cases, but easier.  Thus $G$ is  fractionally $(\Delta-\epsilon)$-colourable, a contradiction.
\end{proof}

\begin{lemma}\label{lem:reduce2}
If $\Delta\geq 6$, $G$ cannot contain $K_\Delta$ minus the edges of vertex disjoint paths, one of length one and one of length two.
\end{lemma}

\begin{proof}
Suppose to the contrary that $\Delta\geq 6$ and $G$ contains a subgraph $X$ on $\Delta$ vertices, with vertices $v_1,v_2,v_3,v_4,v_5\in V(X)$ such that the non-edges of $G[X]$ are exactly $\{v_1v_2, v_1v_3, v_4v_5\}$.  We begin with a fractional $(\Delta-\epsilon)$-colouring $\kappa$ of $G'=G\setminus X$ and extend it to a $(\Delta-\epsilon)$-colouring of $G$ as follows.

First we give $v_1$ and $v_2$ weight $\frac 12$ of common colour, leaving them only partially coloured.  This is possible because $v_1$ and $v_2$ have, in total, at most $5\leq \Delta-\epsilon-\frac 12$ coloured neighbours in $G\setminus X$.  Next we give $v_4$ and $v_5$ the same colour, which is possible because at this point the weight of colour on their neighbourhoods totals at most $2+2+\frac 12 \leq \Delta-\epsilon-1$, since they are both adjacent to $v_1$ and $v_2$.  Next we extend $\kappa$ to complete the colouring of $v_1$, $v_2$, $v_3$, $v_4$, and $v_5$ greedily, which we can do since each of these vertices has at least $\frac 12$ weight of repeated colour in its neighbourhood, and at least one uncoloured neighbour in $X$.  Finally we extend greedily to the remaining vertices of $X$, which we can do since each such vertex is adjacent to $v_1$, $v_2$, $v_4$, and $v_5$, and therefore has repeated colour of weight at least $\frac 32$ in its neighbourhood.  Thus $G$ is fractionally $(\Delta-\epsilon)$-colourable, a contradiction.
\end{proof}

\begin{lemma}\label{lem:reduce3}
If $\Delta\geq 7$, $G$ cannot contain $K_\Delta$ minus the edges of two vertex-disjoint paths of length two.
\end{lemma}

\begin{proof}
Suppose to the contrary that $\Delta\geq 7$ and $G$ contains a subgraph $X$ on $\Delta$ vertices, with vertices $v_1,\dots, v_6\in V(X)$ such that the non-edges of $G[X]$ are exactly $\{v_1v_2,v_2v_3,v_4v_5,v_5v_6\}$.  We begin with a fractional $(\Delta-\epsilon)$-colouring $\kappa$ of $G'=G\setminus X$ and extend it to a $(\Delta-\epsilon)$-colouring of $G$ as follows.

First we give $v_1$ and $v_2$ the same colour.  Next we give $v_4$ and $v_5$ weight $\frac 12$ of common colour.  We then extend greedily to complete the colouring of $v_3$, $v_4$, $v_5$, and $v_6$, then extend greedily to complete the colouring of $G$.  We can do this because, similar to Lemma \ref{lem:reduce1}, $v_1$ and $v_2$ together have at most 5 neighbours in $G\setminus X$, as do $v_4$ and $v_5$.
\end{proof}

\begin{lemma}\label{lem:reduce4}
If $\Delta\geq 7$, $G$ cannot contain $K_\Delta$ minus the edges of a three-edge path.
\end{lemma}

\begin{proof}
Suppose to the contrary that $\Delta\geq 7$ and $G$ contains a subgraph $X$ on $\Delta$ vertices, with vertices $v_1,v_2,v_3,v_4\in V(X)$ such that the non-edges of $G[X]$ are exactly $\{v_1v_2,v_2v_3,v_3v_4\}$.  We begin with a fractional $(\Delta-\epsilon)$-colouring $\kappa$ of $G'=G\setminus X$ and extend it to a $(\Delta-\epsilon)$-colouring of $G$ as follows.

We first extend $\kappa$ by colouring $v_1$ and $v_2$ with the same set of colours. Since $v_1$ has at most two coloured neighbours in $G'$ and $v_2$ has at most three coloured neighbours, we have $|\alpha(v_1)\cap \alpha(v_2)| \geq (\Delta-\epsilon-5) \geq 2-\epsilon \geq 1$, and so from this set we choose $\kappa[v_1]=\kappa[v_2]$.

We next extend $\kappa$ by giving $v_3$ and $v_4$ weight $\frac 12$ of common colour, which is possible because $v_3$ and $v_4$ together have at most 5 neighbours in $G\setminus X$, and weight 1 of colour appearing in their neighbourhood in $X$.  We may then extend greedily to complete the colouring of $v_3$ and $v_4$.  Now since the weight of colours appearing twice in $X$ is at least $\frac 32$, we may extend the colouring to the rest of $X$ greedily.  Thus $G$ is fractionally $(\Delta-\epsilon)$-colourable, a contradiction.
\end{proof}

We are now ready to prove Lemma \ref{lem:mce2}.

\begin{proof}[Proof of Lemma \ref{lem:mce2}]
Suppose $G$ contains a clique $C$ of size $\Delta-1$ and a vertex $w$ outside $C$ with at least two neighbours in $C$.  Call the vertices in $C$ $v_1,\ldots,v_\omega$, and suppose $w$ is adjacent to $v_1$ and $v_2$.  
Let the neighbours of $v_1$ and $v_2$ outside $C\cup \{w\}$ be denoted $y$ and $z$, if they exist.  We may actually assume they exist, since adding them as pendant vertices does not affect our proof adversely.  

We choose $w$, $v_1$, and $v_2$ such that if possible, $w$ is in a $K_\omega$, and subject to that, if possible, $v_1$ and $v_2$ do not have a common neighbour outside $C\cup \{w\}$, i.e.\ $y\neq z$.  We construct one of two reduced graphs from $G$, depending on whether or not $y$ and $z$ are distinct.

\vspace{.5em}\noindent{\bf Case 1:} $y \neq z$.

Let $p$ and $p'$ be the neighbours of $v_3$ outside $C$.  Subject to whether or not we can choose $w$ to be in a $K_\omega$ and whether or not we can choose $v_1$ and $v_2$ such that $y\neq z$, we choose $w$, $v_1$, $v_2$, and $v_3$ such that $w$ and $v_3$ are nonadjacent and $|\{p,p'\}\cap \{y,z\}|$ is minimum. Choose $v_4$ nonadjacent to $w$ as well, noting that this is possible since by Lemma \ref{lem:disjoint}, $w$ has at least two non-neighbours in $C$. Construct the graph $G_1$ from $G-C$ by making $y$ adjacent to $z$ and making $w$ adjacent to $p$ and $p'$.  Clearly $\Delta(G_1)\leq \Delta$.

We claim that $G_1$ is not fractionally $(\Delta-\epsilon)$-colourable; if it is then we extend a $(\Delta-\epsilon)$-colouring $\kappa$ of $G_1$ to a colouring of $G$ as follows. First, we extend $\kappa$ by giving $v_3$ the same colours as $w$. Since all the coloured neighbours of $v_3$ are adjacent to $w$ in $G_1$, we have $\kappa[w]\subseteq \alpha(v_3)$, and so we may choose $\kappa[v_3]=\kappa[w]$. We now greedily extend to the vertices $v_4,\dots,v_{\omega}$, which is possible because $v_1$ and $v_2$ remain uncoloured; it now remains to colour $v_1$ and $v_2$. Since $\kappa[v_3]=\kappa[w]$, it follows that $|\alpha(v_1)| \geq (\Delta-\epsilon)-(\Delta-3)-1 \geq 2-\epsilon$ and $|\alpha(v_2)|\geq 2-\epsilon$. Further, since $|\kappa[y]\cap\kappa[z]|=0$ we have $|\kappa[N(v_1)]\cap\kappa[N(v_2)]| \leq \Delta-3$, and so $|\alpha(v_1)\cup\alpha(v_2)| \geq 2$. Thus we may apply Lemma \ref{lem:hall} to extend $\kappa$ to $v_1$ and $v_2$. It follows that $G$ is fractionally $(\Delta-\epsilon)$-colourable, a contradiction.
This proves the claim. 

Therefore by the minimality of $G$ we may assume that either $G_1$ contains a $\Delta$-clique, or $\Delta=5$ and $G_1$ contains a copy of $C_5\boxtimes K_2$.

We claim that if $\Delta=5$, $G_1$ does not contain a copy $X$ of $C_5\boxtimes K_2$. Suppose to the contrary that adding the edges $wp,wp',yz$ to $G$ yields a copy of $C_5\boxtimes K_2$. Since $G$ does not contain two intersecting copies of $K_4$, $X$ contains two disjoint edges that are not edges of $G$. It follows that $w,y,z\in V(X)$. Further, since $C_5\boxtimes K_2$ is $5$-regular, $wp$ and $wp'$ both belong to $E(X)$ and further no vertex in $V(X)$ has a neighbour in $G\setminus(X\cup\{v_1,v_2,v_3\})$. Therefore $\{v_1,v_2,v_3\}$ is a clique cutset of size three, contradicting the fact that every proper induced subgraph of $G$ is fractionally $(\Delta-\epsilon)$-colourable. This proves the claim.

We may now move on to the more complicated task of proving that $\omega(G_1)=\omega$.  Suppose $G_1$ contains a $\Delta$-clique $C'$.

\textbf{Our first claim} is that $\{w,y,z\}\in C'$ and $yz\notin E(G)$. By Lemma \ref{lem:disjoint}, adding a single edge to $G$ cannot create a $\Delta$-clique. It follows that $w\in V(C')$. Suppose that $|\{y,z\}\cap C'|\leq 1$ or that $yz\in E(G)$. Again by Lemma \ref{lem:disjoint}, $p$, $p'$ must be distinct and belong to $C'$. Now, in $G$, $w$ has $\omega-2$ neighbours in the $\omega$-clique $C'- w$, and so $w$ does not belong to an $\omega$-clique by Lemma \ref{lem:disjoint}. On the other hand, $v_3$ has two neighbours in a $\omega$-clique (namely $p$ and $p'$) and does belong to a maximum clique, contradicting our choice of $w$. This proves the first claim.

\textbf{Our second claim} is that $|\{p,p'\}\cap \{y,z\}|=1$.
Suppose $|\{p,p'\}\cap \{y,z\}|=0$.  Then the edges in $\{wp,wp',yz\}\setminus E(G)$ either consist of a single edge, a two-edge matching, or a $2$-edge path disjoint from a third edge.  By Lemmas \ref{lem:disjoint}, \ref{lem:reduce1}, and \ref{lem:reduce2}, we know that they consist of a $2$-edge path disjoint from a third edge, and that $\Delta=5$.  In particular, it follows from the first claim that $w$ is adjacent to both $p$ and $p'$.  Let $p''$ and $p'''$ denote the neighbours of $v_4$ outside $C$.  Since $G$ does not contain a bump by Lemma \ref{lem:bump}, both $y$ and $z$ have only one neighbour in $C$.  We may therefore exchange the roles of $v_3$ and $v_4$ without violating the disjointness of $\{p,p'\},\{y,z\}$.  By the minimality of $G$, the new resulting reduced graph $G_1'$ (constructed as was $G_1$, but with $v_3$ and $v_4$ swapped) has a $K_\Delta$.  Since $y$ is adjacent to $w$, $v_1$, $p$, $p'$, $p''$, and $p'''$, the sets $\{p,p'\}$ and $\{p'',p'''\}$ must intersect.  Since $\{p,p',v_3,v_4\}$ cannot be a clique by Lemma \ref{lem:disjoint}, $\{p,p'\} \neq \{p'',p'''\}$.  Therefore we may assume $p' = p'''$ and that $|\{p,p',p''\}|=3$.  But then $p'$ is adjacent to $p$, $p''$, $y$, $z$, $v_3$, and $v_4$, contradicting the fact that $\Delta=5$.  Therefore $|\{p,p'\}\cap \{y,z\}|\neq 0$.

Suppose $|\{p,p'\}\cap \{y,z\}|=2$. We may assume $p=y$ and $p'=z$. Recall that we have chosen $v_3$ so as to minimize $|\{p,p'\}\cap \{y,z\}|$. 
Each of $w,y,z$ has at most three neighbours in $C$, since $\{w,y,z\}\in C'$ by the first claim.  Furthermore, if $w$ belongs to a $K_\omega$ in $G$, then it has only two neighbours in $C$. If $\Delta=5$, by Lemma \ref{lem:disjoint}, $v_4$ sees neither $y$ nor $z$ (since $v_3$ sees $y$ and $z$), contradicting our choice of $v_3$.  If $\Delta\geq 6$ and $w$ is in a $K_\omega$ in $G$, then there is a vertex in $C\setminus N(w)$ that is adjacent to at most one of $y,z$ and nonadjacent to $w$, contradicting our choice of $v_3$.  If $\Delta\geq 6$ and $w$ is not in a $K_\omega$ in $G$, then either there is a vertex in $C\setminus N(w)$ adjacent to at most one of $y,z$, contradicting our choice of $v_3$, or else every vertex in $C\setminus N(w)$ sees both of $y,z$.  In this latter case we can relabel: relabel $y$ to $w'$, $v_1$ to $v_1'$, $v_3$ to $v_2'$, $w$ to $y'$, $z$ to $z'$, and $v_4$ to $v_3'$. Since $v_4$ was chosen to be nonadjacent to $w$, we have a labelling that contradicts the minimality of $|\{p,p'\}\cap \{y,z\}|$.
This proves the second claim. We may now assume that $y=p$ and that $|\{y,z,p'\}|=3$.

\textbf{Our third claim} is that $p'\in C'$.
Suppose to the contrary that $p'\notin C'$. Then in $G$, $w$ has $\omega-1$ neighbours in $V(C')$. Thus $w$ belongs to an $\omega$-clique in $G-C$, and therefore has exactly two neighbours in $C$. Also, since $wz\in E(G)$ and $G$ does not contain a bump by Lemma \ref{lem:bump}, $v_2$ is the only neighbour of $z$ in $C$. Further, $y$ belongs to an $(\omega-1)$-clique in $G-C$ and has at most three neighbours in $C$, and at most two if $\Delta=5$. Therefore, there is a vertex in $C$ with no neighbour in $\{w,y,z\}$, contradicting our choice of $v_3$.
This proves the third claim.

We now know that $y=p$ and $\{w,y,p',z\}\subseteq V(C')$.
Since $G$ does not contain a bump and since $wz\in E(G)$, we know that $z$ has only one neighbour in $C$.  Therefore by our choice of $v_3$ minimizing $|\{p,p'\}\cap\{y,z\}|$, every vertex in $C$ is adjacent to $w$ or $y$.  Thus $\Delta=6$ and each of $w$ and $y$ has three neighbours in $C$.

To complete the proof, we now fractionally colour $G$ directly, beginning with a fractional $(\Delta-\epsilon)$-colouring $\kappa$ of $G-C-\{w,y\}$.  
We first extend $\kappa$ by colouring $w$ and $v_3$ with the same set of colours. Since $v_3$ and $w$ together have at most four coloured neighbours, we have $|\alpha(v_3)\cap \alpha(w)| \geq (6-\epsilon)-4 \geq 1$, and so we may choose $\kappa[v_3]=\kappa[w]$.

Next we extend $\kappa$ by colouring $y$ and $v_2$ so that $|\kappa[y]\cap\kappa[v_2]|\geq \frac 12$, which is possible because at this point, $|\kappa[N(y)\cup N(v_2)]|=5$, since the only coloured vertices in $N(y)\cup N(v_2)$ are $C'-y$ and $v_3$ (which has the same colour as $w$). We now have $\kappa[\{v_2,v_3,w,y\}]\leq \tfrac 52$.

Next we extend $\kappa$ by colouring $v_4$ and $v_5$. Since each of $v_4,v_5$ is adjacent to either $w$ or $y$, we have $|\kappa[N(v_4)]|\leq \tfrac 72$ and $|\kappa[N(v_5)]|\leq \tfrac 72$. Thus $|\alpha(v_4)|, |\alpha(v_5)|\geq 2$ and so we may apply Lemma \ref{lem:hall} to choose $\kappa[v_4]$ and $\kappa[v_5]$ greedily.

Finally we greedily extend $\kappa$ to $v_1$. We have $\kappa[N(v_1)] \leq \tfrac 92$ since $v_1$ is adjacent to $v_2$, $v_3$, $w$, and $y$. Applying Lemma \ref{lem:hall}, we may choose $\kappa[v_1]$ from $\alpha(v_1)$. Thus $G$ is fractionally $(\Delta-\epsilon)$-colourable, a contradiction.

{\bf This completes the proof of Case 1.}

\vspace{.5em}\noindent{\bf Case 2:} $y = z$ and $w$ is in a $K_\omega$ in $G$.

In this case, we know that we can choose $w$ to be in a maximum clique, but we cannot make such a choice of $w,v_1,v_2$ for which $y\neq z$.  Since $w$ is in a maximum clique, it has only two neighbours in $C$.  Therefore we may choose $v_3$ and $v_4$ to be nonadjacent to both $w$ and $y$, since Lemma \ref{lem:disjoint} implies that $y$ has at least two non-neighbours in $C$.  But we need further conditions on our vertex labelling.  Denote by $p,p'$ and $q,q'$ the neighbours of $v_3$ and $v_4$ outside $C$, respectively.
We choose a labelling of the vertices satisfying the following conditions:
\begin{enumerate*}
\item[$L1$] $w$ is in a maximum clique.  Subject to this condition,
\item[$L2$] $y$ is in a maximum clique if possible.  Subject to this condition,
\item[$L3$] $v_3$ and $v_4$ are not adjacent to $w$ nor to $y$. Subject to satisfying the previous conditions,
\item[$L4$] $v_3$ is chosen so that $|N(p)\cap N(p')\cap N(y)|$ is maximized.
\end{enumerate*}

Construct the graph $G_2$ from $G-C$ by making $w$ adjacent to $p$ and $p'$ and making $y$ adjacent to $q$ and $q'$. Clearly $\Delta(G_1)\leq \Delta$.

We claim that $G_2$ is not fractionally $(\Delta-\epsilon)$-colourable; if it is then we extend a $(\Delta-\epsilon)$-colouring $\kappa$ of $G_2$ to a colouring of $G$ as follows.
We begin by extending $\kappa$ to colour $v_3$ with the same colour as $w$. Since $v_3$'s only coloured neighbours are $p$ and $p'$,  which are adjacent to $w$ in $G_2$, we may choose $\kappa[v_3]=\kappa[w]$.
We now extend $\kappa$ to the remaining vertices in $C$.
By the choice of $\kappa[v_3]$, we have $|\alpha(v_1)|, |\alpha(v_2)| \geq \Delta-\epsilon-2$. Since each of the $\Delta-4$ other uncoloured vertices has at most three coloured neighbours we find $|\alpha(v_i)|\geq \Delta-\epsilon-3$ for $4\leq i \leq \omega$. Third, the edges $yq,yq'$ in $G_2$ ensure that $|\alpha(v_1)\cup \alpha(v_4)|, |\alpha(v_2)\cup \alpha(v_4)| \geq \Delta-\epsilon-1$. Applying Lemma \ref{lem:hall} to $C\setminus\{v_3\}$ (which has size $\Delta-2$), we find a $(\Delta-\epsilon)$-colouring of $G$, a contradiction.
This proves the claim.

Therefore by the minimality of $G$ we may assume that either $G_2$ contains a $\Delta$-clique, or $\Delta=5$ and $G_2$ contains a copy of $C_5\boxtimes K_2$.
Let $F=E(G_2)\setminus E(G) \subseteq \{wp,wp',yq,yq'\}$. Let $F_w$ and $F_y$ denote the edges incident to $w$ and $y$ in $G_2$, respectively.

We claim that if $\Delta=5$, $G_2$ does not contain a copy $X$ of $C_5\boxtimes K_2$. Suppose to the contrary that adding the edges $wp,wp',yq,yq'$ to $G$ creates a copy of $C_5\boxtimes K_2$. Since $G$ does not contain two intersecting copies of $K_4$, $X$ contains at least two vertex-disjoint edges that are not edges of $G$. It follows that $w,y\in V(X)$. Further, since $C_5\boxtimes K_2$ is $5$-regular, $\{p,p',q,q'\}\subseteq V(X)$ and $F$ contains all four edges $wp,wp',yq,yq'$. Since $w$ belongs to a $K_4$ in $G$, $p$ and $p'$ must form the intersection of two $K_4$s in $X$. Since $G$ does not contain a pair of intersecting $K_4$s, $q$ and $q'$ do not form the intersection of two $K_4$s in $X$, and moreover, $y$ cannot be in $N(w)\cup N(p) \cup N(p')$ in $X$.  Hence $y$ does not belong to a $4$-clique in $G$.  See Figure \ref{fig:case2}, where $y$ is the bottom left vertex.
Observe that $v_3$ belongs to a maximum clique in $G$, and its neighbours $p$ and $p'$ belong to another maximum clique. Further, $p$ and $p'$ have a common neighbour in a third maximum clique.  Since $y$ is not in a maximum clique, this contradicts $L2$ in our choice of $w$ and $y$, and proves the claim.

\begin{figure}
\begin{center}
\includegraphics[scale=.7]{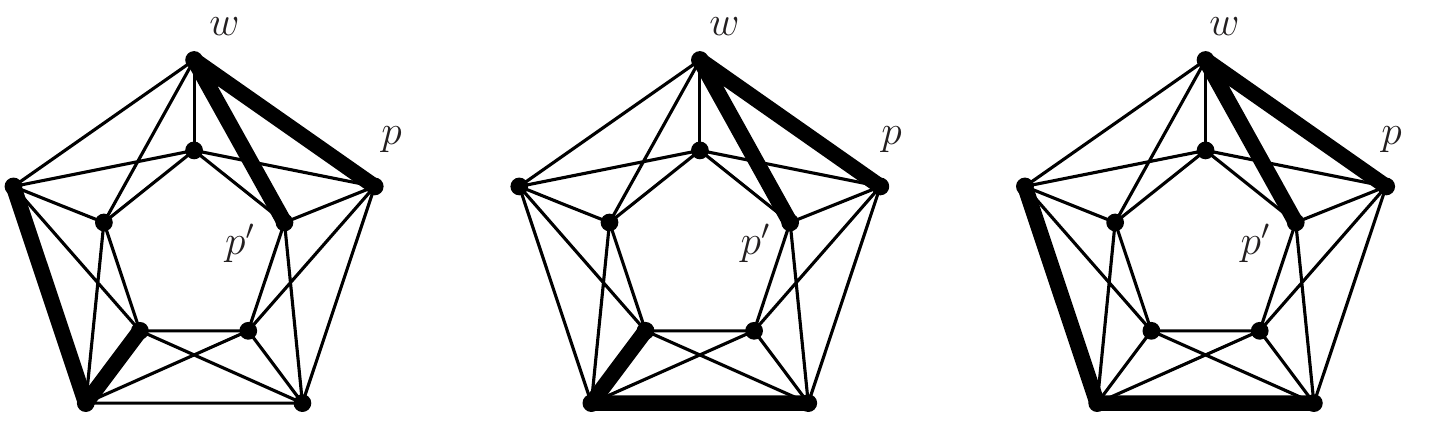}
\end{center}\caption{Three ways to form $C_5\boxtimes K_2$ in Case 2.}\label{fig:case2}
\end{figure}

We now move on to the task of proving that $\omega(G_2)=\omega$.  Suppose $G_2$ contains a $\Delta$-clique $C'$. 

\textbf{Our first claim} is that  $|E(C')\cap F_y|\geq 1$ and $|E(C')\cap F_w|\geq 1$.  We can see that $|E(C')\cap F_y|\geq 1$, otherwise $C'\setminus w$ is a maximum clique in $G$ intersecting a maximum clique containing $w$, contradicting Lemma \ref{lem:disjoint}.

Suppose now that $|E(C')\cap F_w|=0$. By the same argument, $y$ cannot belong to a maximum clique in $G$.  We know that $C'$ must contain at least two edges in $F$, so $yq,yq'\in F\cap E(C')$. Therefore $|N(q)\cap N(q')\cap N(y)|\geq \omega-2$ and these vertices, along with $y$ form an $(\omega-1)$-clique. Further $qq'\in E(G)$, and so $q,q'$ belong to an $\omega$-clique in $G$.

If $|\{p,p'\}\cap \{q,q'\}| = 1$, then we can relabel $v_4$ as $w'$; since $v_4$ is in a $K_\omega$ in $G$ and has two neighbours in $C'\setminus\{y\}$, one but not both of which are adjacent to $v_3$, contradicting the fact that we are not in Case 1.  If $|\{p,p'\}\cap \{q,q'\}| = 2$, this contradicts condition $L2$ in our choice of labelling, since $G$ contains two vertices in the $K_\omega$ $C$ having two neighbours in common in a disjoint $K_\omega$ $C'\setminus\{y\}$.  Therefore $|\{p,p'\}\cap \{q,q'\}| = 0$.

Note that by $L3$ we have chosen $v_3$, $v_4$ nonadjacent to both $w$ and $y$. In particular this means that $\{p,p'\}$ and $\{q,q',y\}$ are disjoint.
By $L4$, we know $|N(p)\cap N(p')\cap N(y)|\geq |N(q)\cap N(q')\cap N(y)| \geq \omega-2$. In particular this set must intersect $N(q)\cap N(q')\cap N(y)$.  But since $\{p,p'\}$ and $\{q,q',y\}$ are disjoint, if $\{p,p'\}\cap C' = \emptyset$, there is a vertex of degree $\Delta+1$, a contradiction.  Therefore we may assume without loss of generality that $p\in C'\setminus\{q,q',y\}$.  But then in $G$, $p$ is adjacent to every other vertex in $C'$, so its only other neighbour is $v_3$.  Since $y$ is nonadjacent to $v_3$, $q$, and $q'$, $N(p)\cap N(p')\cap N(y)\subseteq C'\setminus\{q,q',y,p\}$, contradicting the fact that its size is at least $\omega-2$.  This proves the first claim.

\textbf{Our second claim} is that $w$ and $y$ belong to an $\omega$-clique $W$ in $G$.  As a consequence, since this makes $\{w,y,v_1,v_2\}$ a clique, Lemma \ref{lem:disjoint} tells us that $\Delta \geq 6$.  To prove this, let $W$ be the maximum clique in $G$ containing $w$, and note that $W$ is the closed neighbourhood of $w$ in $G-C$.  By the first claim, $w\in V(C')$ and $y\in V(C')$.  By the choice of $v_3$, $y\notin\{p,p'\}$ and $w\notin \{q,q'\}$. It follows that $wy\in E(G)$, and so $y\in W$. This proves the second claim.

\textbf{Our third claim} is that the only edges between $C$ and $W$ are between $\{v_1,v_2\}$ and $\{w,y\}$.  To see this assume otherwise, and denote the vertices of $W$ $\{w,y,w_3,\ldots, w_\omega\}$.  By the maximum degree, there must exist $3\leq i,j\leq \omega$ such that $v_i$ and $w_j$ are adjacent.

To reach a contradiction we extend a fractional $(\Delta-\epsilon)$-colouring $\kappa$ of $G-W-C$ as follows.  First assign $w$ and $v_i$ the same colour, which is possible because together these vertices have at most weight $1$ of colour on (the union of) their neighbourhoods.  Then for some $i' \notin \{1,2,i\}$, give $y$ and $v_{i'}$ colour $\frac 12$ in common, leaving them only partially coloured, noting that this is possible because at this point $y$ and $v_{i'}$ have colour at most $1+2=3$ on their neighbourhoods (since $w$ and $v_i$ have the same colour).  Next we greedily extend to all vertices of $W\setminus \{w,y,w_j\}$, noting that this is possible because all these vertices are adjacent to $y$ and $w_j$, which together have only $\frac 12$ colour on them at this point.  We then greedily extend to $w_j$, which is possible because $w_j$ is adjacent to $w$, $y$, and $v_i$, which together have weight $\frac 32$ colour on them.  Next we greedily extend to complete the colouring of all vertices of $(C\cup \{y\}) \setminus \{v_1,v_2\}$, which is clearly possible because $v_1$ and $v_2$ are still uncoloured.  Finally we extend to $v_1$ and $v_2$, which is possible because both are complete to $\{w,y,v_i,v_{i'}\}$, a set of four vertices with at most $\frac 52$ colour on them.  This contradicts the fact that $G$ is not fractionally $(\Delta-\epsilon)$-colourable, and proves the third claim.

\textbf{Our fourth claim} is that $\{p,p'\}\cap \{q,q'\} \neq \emptyset$.  By the second claim, neither $w$ nor $y$ has any neighbours outside of $W$ in $G-C$.  By the first claim, $C'$ contains an edge in $F_w$ and an edge in $F_y$; we may assume without loss of generality that $p\in V(C')$.  By the third claim $\{p,p',q,q'\}\cap W = \emptyset$, so $G$ contains no edges between $\{w,y\}$ and $\{p,p',q,q'\}$.  Therefore since $C'$ is a clique in $G_2$, $yp$ must be in $F$, so $p\in\{q,q'\}$.  This proves the fourth claim.

Without loss of generality, for the remainder of Case 2 we assume $p\in V(C')$ and $p=q$. Thus we can also assume that $p$ is adjacent to $w_3$ and $w_4$ in $W$. By the third claim $p$ does not belong to $W$.

We now complete the proof of Case 2. To do so we fractionally colour $G$ by extending a $(\Delta-\epsilon)$-colouring $\kappa$ of $G-W-C-\{p\}$ as follows.
We begin to extend $\kappa$ by assigning $\kappa[w]=\kappa[v_5]$, noting that $v_5$ may or may not be adjacent to $p$. This is possible since together these vertices have at most two coloured neighbours.
Next we give $v_1$ and $w_5$ colour $\tfrac 12$ in common, leaving them partially uncoloured. This is possible since at this point $|\kappa[N(v_1)\cup N(w_5)]|\leq 3$.
Next, we extend $\kappa$ by giving $w_4$ and $v_4$ the same set of colours, noting that since both are adjacent to $p$, at this point at most $\frac 72\leq \Delta-\epsilon-1$ colour appears on their neighbourhoods, so this is possible.
Next, we give $w_3$ and $v_3$ common colour $\tfrac 12$, noting that both are adjacent to $p$. Since $\kappa[N(v_3)\cap C]=\kappa[(N(w_3)\cap W)\cup\{v_3\}]$ and $|\kappa[N(v_3)\cap C]|=\tfrac 52$, we have $|\alpha(v_3)\cap\alpha(w_3)| \geq (\Delta-\epsilon)-\tfrac 52 - 2 \geq 1$.
We now greedily extend $\kappa$ to colour $W-\{w,y,w_3,w_4\}$, which is possible since $y$ and $w_3$ together have weight $\frac 32$ not yet coloured.  Next we give $y$ and $v_3$ weight $\frac 12$ of colour in common and leave them partially uncoloured, which is possible because at this point $|\alpha(y)|\geq \frac 32$, and $|\kappa[\tilde N(v_3)]\setminus \kappa[\tilde N(y)]| \leq 1$.  We can now greedily extend to $C-\{v_1,\dots,v_5\}$, since $v_1$ and $v_2$ together have weight $\frac 32$ not yet coloured.  Next we can extend to complete the colouring of $w_3$, since $y$ and $p$ together have weight $\frac 32$ not yet coloured.  Next we can extend to complete the colouring of $y$, since $v_1$ and $v_2$ together have weight $\frac 32$ not yet coloured.

Finally we can complete the colouring by extending greedily to complete the colouring of $v_1$ and $v_2$, since each has weight at least $\frac 32$ of colour appearing twice on its neighbourhood.  This completes the proof of Case 2.

{\bf This completes the proof of Case 2.}

\vspace{.5em}\noindent{\bf Case 3:} $y = z$ and $w$ is not in a $K_\omega$ in $G$.

In this case, by the choice of $w$, there exists no vertex in $G$ belonging to a maximum clique that has two neighbours in a different maximum clique. Also, we know that every pair of vertices in $C$ has either zero or two common neighbours outside of $C$, for otherwise with a better choice of $w, v_1,v_2$ we would be in Case 1. Thus $N(w)\cap V(C) = N(y)\cap V(C)$. By Lemma \ref{lem:disjoint}, $|V(C)\setminus N(w)| \geq 2$.
Again denote by $p,p'$ and $q,q'$ the neighbours of $v_3$ and $v_4$ outside $C$, respectively. We choose $v_3$ and $v_4$ from $V(C)\setminus N(w)$ to maximize $|\{p,p',q,q'\}|$. Subject to this, $v_3$ and $v_4$ are chosen to maximize $|\{wp,wp',yq,yq'\}\cap E(G)|$. 
Note that $|\{p,p'\}\cap \{q,q'\}|\in \{0,2\}$, that $\{w,y\} \cap \{p,p',q,q'\} = \emptyset$, and that in particular, $y$ is nonadjacent to $v_4$.

Noting that $w,y \notin \{p,p',q,q'\}$, we construct the graph $G_2$ from $G-C$ as in Case 2 by making $w$ adjacent to $p$ and $p'$ and making $y$ adjacent to $q$ and $q'$.
As in Case 2, we may assume $G_2$ is not fractionally $(\Delta-\epsilon)$-colourable; if it is then we extend a $(\Delta-\epsilon)$-colouring $\kappa$ of $G_2$ to a colouring of $G$. (Observe that the colouring argument given in Case 2 does not make use of the fact that $w$ belongs to a maximum clique in that case.)

Therefore we may assume that either $G_2$ contains a $\Delta$-clique, or $\Delta=5$ and $G_2$ contains a copy of $C_5\boxtimes K_2$. As in the previous case, let $F=E(G_2)\setminus E(G) \subseteq \{wp,wp',yq,yq'\}$. Let $F_w$ and $F_y$ denote the edges of $F$ incident to $w$ and $y$ in $G_2$, respectively.

We claim that if $\Delta=5$, $G_2$ does not contain a copy $X$ of $C_5\boxtimes K_2$. Suppose to the contrary that adding the edges $wp,wp',yq,yq'$ to $G$ creates a copy of $C_5\boxtimes K_2$. Since $G$ does not contain two intersecting copies of $K_4$, $X$ contains two vertex-disjoint edges of $F$.  It follows that $w,y\in V(X)$, and since $\Delta=5$, Lemma \ref{lem:disjoint} tells us that $w$ and $y$ are not adjacent. Further, since $C_5\boxtimes K_2$ is $5$-regular, $\{p,p',q,q'\}\subseteq V(X)$ and $F$ contains all four edges $wp,wp',yq,yq'$.  Since $w$ does not belong to a $K_4$ in $G$, $p$ and $p'$ do not form the intersection of two $K_4$s in $X$.  Likewise, neither do $q$ and $q'$.  Also, if $\{p,p'\}\cap \{q,q'\}\neq \emptyset$ then $|\{p,p'\}\cap \{q,q'\}|=2$ (since we are not in Case 1), which is impossible because intersection of the neighbourhoods of two nonadjacent vertices in $C_5\boxtimes K_2$ is the intersection of two $K_4$s, a contradiction.  Therefore $w,y,p,p',q,q'$ are six distinct vertices.

\begin{figure}[h]
\begin{center}
\includegraphics[scale=.4]{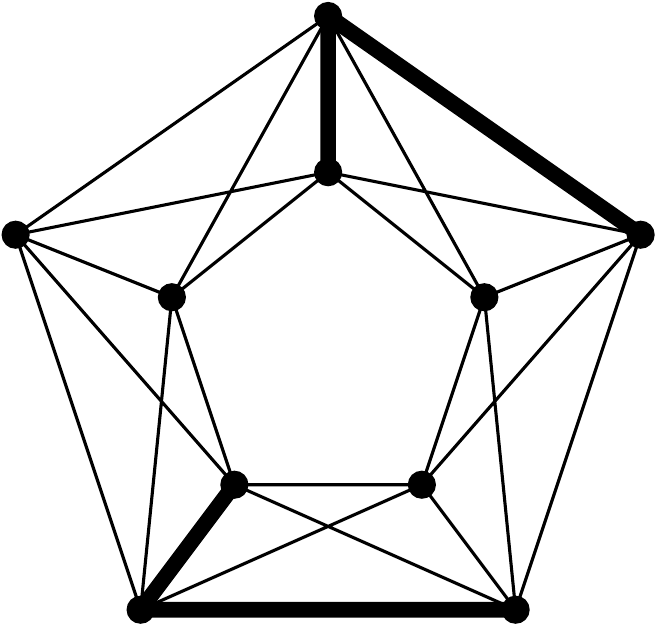}
\end{center}\caption{The only way to form $C_5\boxtimes K_2$ in Case 3.  If $w$ is the top vertex, we may instead choose $w$ as the vertex immediately below it to put us in Case 1.}\label{fig:case3}
\end{figure}

Since exchanging the roles of $v_3$ and $v_4$ cannot reduce $|F|$, $G$ contains no edges from $\{wq, wq', yp, yp'\}$.  It follows that $pp'\in E(G)$ and $qq'\in E(G)$.  Therefore by symmetry, bearing in mind that $w$ and $y$ are nonadjacent in both $G$ and $G_2$, the only possible case is shown in Figure \ref{fig:case3}.  Note here that there is a different choice of $w$ that would put us in Case 1, a contradiction.

We now proceed to prove that $\omega(G_2)<\Delta$.  Suppose $G_2$ contains a $\Delta$-clique $C'$. 

\textbf{Our first claim} is that $|E(C')\cap F_w|\geq 1$ and $|E(C')\cap F_y|\geq 1$.
Suppose that $|E(C')\cap F_w|=0$. Then clearly $y\in V(C')$, and by Lemma \ref{lem:disjoint}, both edges $yq,yq'$ belong to $E(C')$, and $qq'\in E(G)$. But then $C'-y$ is an $\omega$-clique containing two neighbours of $v_4$, which also belongs to an $\omega$-clique. This contradicts our choice of $w$.
By a symmetric argument, $|E(C')\cap F_y|\geq 1$.
This proves the first claim.

\textbf{Our second claim} is that $wy\in E(G)$ and $\Delta \geq 6$.
By the first claim, $w$ and $y$ belong to $V(C')$. By the choice of $v_3$ and $v_4$, $w,y \notin \{p,p',q,q'\}$. Thus $wy\in E(G)$, and so $w,y,v_1,v_2$ form a $K_4$. If $\Delta=5$ this contradicts Lemma \ref{lem:disjoint}. This proves the second claim.

\textbf{Our third claim} is that $|E(C')\cap F|\geq 3$.
Suppose that $|E(C')\cap F| =2$. By Lemma \ref{lem:reduce1}, the two edges in $E(C')\cap F$ do not form a matching, and so they form a two-edge path.  By the first claim, one of the edges must be between $w$ and $y$, contradicting the second claim.
This proves the third claim.

\textbf{Our fourth claim} is that $|E(C')\cap F|=4$.
Suppose that $|E(C')\cap F| =3$. By Lemma \ref{lem:reduce2}, at least two pairs of the edges in $E(C')\cap F$ intersect. 
Since $w,y \notin \{p,p',q,q'\}$ the edges $E(C')\cap F$ do not form a triangle, so they form a three-edge path.
By Lemma \ref{lem:reduce4} and the second claim, $\Delta= 6$. 

Since $wy\in E(G)$ and by symmetry between $w$ and $y$ and between $p$ and $p'$, we may assume $p=q$ and the edges of the path are $p'w,wp,py$.  Since $|\{p,p'\}\cap \{q,q'\}|\neq 1$, $p'=q'$ and $pp'\in E(G)$. By the choice of $v_3,v_4$ maximizing $|\{p,p',q,q'\}|$, $v_5$ must be complete to $\{p,p'\}$ or to $\{w,y\}$. But then $v_5$ belongs to two $5$-cliques in $G$, contradicting Lemma \ref{lem:disjoint}.
This proves the fourth claim.

We now know that $|E(C')\cap F|=4$. 
Suppose that the edges in $E(C')\cap F$ form two vertex-disjoint two-edge paths. Then by Lemma \ref{lem:reduce3}, $\Delta=6$. Now $|\{p,p',q,q'\}|=4$ and so $wq,wq',yp,yp' \in E(G)$. This contradicts the choice of $v_3$ and $v_4$, for reversing their roles would yield $|F|=0$.

Since we are in Case 3, the edges in $E(C')\cap F$ therefore form a cycle of length four. It follows that $\{p,p'\}=\{q,q'\}$ and $wy,pp'\in E(G)$. By the choice of $v_3$ and $v_4$ maximizing $|\{p,p',q,q'\}|$, each of $v_5,\dots,v_{\omega}$ is complete to either $\{w,y\}$ or $\{p,p'\}$.  Therefore by Lemma \ref{lem:disjoint}, $\Delta \geq 7$.  Since each of $w,y,p,p'$ is adjacent to $\Delta-3$ vertices of $C'$ in $G$, each has at most three neighbours in $C$.  Therefore $\Delta=7$, and $G$ is isomorphic to the graph $(C_5\boxtimes K_3) - 2v$ pictured in Figure \ref{fig:upperbounds}.  Thus $G$ is indeed fractionally $\tfrac{13}{2}$-colourable and thus fractionally $(\Delta-\epsilon)$-colourable, a contradiction.

This completes the proof of Case 3, and the proof of the lemma.
\end{proof}

\section{Future directions}

We have already given several open problems that are worthy of consideration, namely Conjectures \ref{con:1} and \ref{con:2}, which propose, respectively, that $f(6)=f(7)=f(8)=\tfrac 12$ and that $f(4)\geq f(3)$.  We conclude the paper with one more conjecture:

\begin{conjecture}\label{con:3}
Let $G$ be a graph with maximum degree $5$ and clique number $4$ such that no two $4$-cliques intersect and such that no vertex outside any maximum clique $C$ has more than one neighbour in $C$.  Then there is a fractional $4$-colouring of the vertices in $4$-cliques such that for any vertex $v$ not in a $4$-clique, $|\alpha(v)|\geq 1$.
\end{conjecture}

If Conjecture \ref{con:3} were to hold, our fractional colouring method could be applied to greater effect.  In particular, we could easily prove that $f(5) \geq 1/11$ and $f(6) \geq 1/8$.  The improvements would be smaller for larger values of $\Delta$.

\section{Acknowledgements}

The authors are very grateful to the two referees for their thorough, helpful, and speedy reviews.

\bibliography{masterbib}
\end{document}